\PassOptionsToPackage{draft}{hyperref}

\documentclass[sn-mathphys-num]{sn-jnl}

\usepackage[T1]{fontenc} 
\usepackage{lmodern}     
\usepackage[scr=rsfs]{mathalfa} 

\usepackage{graphicx}%
\usepackage{multirow}%
\usepackage{amsmath,amssymb,amsfonts}%
\usepackage{amsthm}%
\usepackage{mathrsfs}%
\usepackage[title]{appendix}%
\usepackage{xcolor}%
\usepackage{textcomp}%
\usepackage{manyfoot}%
\usepackage{booktabs}%
\usepackage{algorithm}%
\usepackage{algorithmicx}%
\usepackage{algpseudocode}%
\usepackage{listings}%
\usepackage{url}%

\usepackage{subcaption}

\newtheorem{theorem}{Theorem}
\newcommand{\Kbold}{\mathbf{K}}
\newcommand{\mubold}{\boldsymbol{\mu}}
\newcommand{\thetabold}{\boldsymbol{\theta}}

\newcommand{\Lbold}{\mathbf{L}}
\newcommand{\Bbold}{\mathbf{B}}

\begin{document}

\title[Ancestor Hawkes]{The Ancestor Hawkes Process with an Application to Group Chat Data}


\author*[1]{\fnm{Gordon} \sur{J. Ross}}\email{gordon.ross@ed.ac.uk}

\author[2]{\fnm{Isabella} \sur{Deutsch}}

\affil*[1]{\orgdiv{Department of Mathematics}, \orgname{University of Edinburgh}, \orgaddress{\street{UK}}}

\affil[2]{\orgname{Danu Insights Limited}, \orgaddress{\city{Edinburgh}, \country{UK}}}


\abstract{The Hawkes process is used to model point process data where events occur in clusters and bursts. In a standard multivariate Hawkes process, every event that occurs in a dimension has an equal impact on the process intensity. However, this assumption is unrealistic in applications such as the modelling of message cascades where the effect of an event depends on whether it was the initiator or a member of a particular cluster. To alleviate this, we introduce a new Hawkes process model, the Ancestor Hawkes process, which allows the impact of each event to vary based on its origin. The relevance of the Ancestor Hawkes process is showcased on real data from a 9-person group chat, where our proposed approach reveals individual response preferences. Crucially, this is achieved in a privacy-conscious manner, as only the sender and the time at which a message was sent -- but not its content -- are utilised. These nuances of messaging cascades are missed by the standard Hawkes process, but are relevant for studying latent interaction structure and for personalised notification management.}




\maketitle

\section{Introduction} 

The Hawkes process \citep{hawkes_spectra_1971} is commonly used to model point process data when events can occur in clusters and bursts. It has enjoyed application in many fields including seismology \citep{ogata_statistical_1988}, retail analytics \citep{Pitkin2024}, and finance \citep{bacry_hawkes_2015}. Given a multivariate setting with $M$ dimensions, suppose that $N$ events are observed on the time interval $[0,T]$. We write these as  $Y = \{y_1,\ldots,y_N\}$ where $y_i = (t_i,d_i)$ and $t_i$ denotes the time of the $i^{th}$ event, with $d_i \in \{1,\ldots,M\}$ denoting  the dimension in which it occurred. The conditional intensity of the multivariate Hawkes process in each dimension $m$ is then given by:

\begin{equation}
       \lambda_m(t \mid \mathcal{H}_t )  = \mu_m(t)+ \sum_{i: t_i < t} K_{d_i \to m} \, g_{d_i \to m} (t-t_i),
       \label{eqn:intro}
\end{equation}
where  $\mathcal{H}_t$ contains all events that occurred before  time $t$. This intensity function has three components: 
\begin{enumerate}
    \item the background rates $\mu_m(\cdot)$ (which can differ across dimensions),
    \item the influence matrix $\mathbf{K} = \{K_{d \to m}\}_{d = 1 \dots M}^{m = 1 \dots M}$ where $K_{d \to m}$ denotes the size of the impact that an event occurring in dimension $d$ has on the intensity function in dimension $m$, and
    \item the influence kernels $g_{d \to m}(\cdot)$ which determine how this impact decays over time.
\end{enumerate}

An important property of the Hawkes process is that it can be interpreted as a branching process \citep{hawkes_cluster_1974} in which the total intensity in each dimension decomposes into an immigrant process with rate $\mu_m(t)$ and offspring processes triggered by past events via $K_{d_i \to m} g_{d_i \to m}(t-t_i)$. This naturally separates events into those which were exogenously generated by the background process (commonly referred to as immigrant events) and those which were triggered by previous events (referred to as triggered events).

The classic Hawkes process assumes that the influence of each event depends only on the dimension in which it occurred, irrespective of the branching structure. However, in some applications, it might be reasonable for immigrant and triggered events to have different influence magnitudes. Therefore, our goal is to motivate and present an extension of the Hawkes process where the influence of each event changes based on whether it is an immigrant or triggered. To that end, we propose the Ancestor Hawkes process\footnote{The term ``Ancestor'' refers here to the branching ancestry of events in a Hawkes cluster, and is unrelated to ancestor sampling in sequential Monte Carlo methods.} which allows immigrant and triggered events to have different influences. Instead of one matrix $\Kbold$ of influence magnitudes for all events as in the classic Hawkes, we use two matrices. Now $\Kbold$ governs the influence magnitude which immigrant events have on all dimensions and $\Lbold$ describes the influence magnitude that triggered events have.


The Ancestor Hawkes process is an extension of models currently utilised in the literature. The closest precursor is \citet{li_brunc_2020}, who estimate a branching structure for Hawkes processes where the branching structure is not available (e.g.\ inhibition, non-linear link function). They then use an adapted Chinese Restaurant Process to allow different triggerings for each cluster of events.  Also somewhat related is \citet{Kolev2019}, who need to introduce a distinction between immigrant and triggered events for the purpose of efficiently estimating a non-Poissonian variant of the Hawkes process. However, their concern is quite different to ours and they do not propose that influence dynamics should differ based on the branching structure. Another important precursor is the ARMA point process introduced by \cite{schatz2022arma} who introduce a self-exciting point process which allows for the first generation of triggered events to come from a different distribution to future generations, although their analysis is restricted to the univariate case. Unlike these approaches, we let the latent immigrant/triggered state itself modulate cross-dimension influence via distinct  $\Kbold$ and $\Lbold$ matrices.

\subsection{Motivation}
To motivate our model, we consider the following extended example of modelling group chat data, which we  further study in Section \ref{sec_groupchat_app} using a real data set.  A group chat is an asynchronous online conversation between more than two participants, commonly conducted over apps such as WhatsApp and Facebook Messenger. We focus on private group chats where only admitted participants can see the conversation and every participant has an equal opportunity to send messages \citep{mannell_plural_2020}. While we use group chats as a running example, our methodological goal is more general: to separate exogenous starts from reply-driven influence in multivariate event streams.

The literature on group chat modelling is sparse. For one of the few examples see \citet{guo_who_2019}, who use a long short-term memory neural network. Classic Hawkes processes have been used to model one-to-one conversations such as emails \citep{miscouridou_modelling_2018}, as well as one-to-many scenarios, like tweets on Twitter/X \citep{rizoiu_tutorial_2017}. However, to our knowledge, they have not been used for small-scale group chats, where all participants contribute actively.

We propose the Ancestor Hawkes model for group chats, as it can capture dynamics in a group chat beyond the classic Hawkes process. The set-up is as follows. Each of the chat participants is represented by a dimension. When they send a message, an event is recorded in that dimension. As a toy example, consider three people A, B, and C  participating in a group chat. We now examine the reply propensity of Person C in two different scenarios with the following messages being sent. 
\begin{enumerate}
    \item Person A: \textit{``Who would like to go out to a restaurant tonight?''} (immigrant event)
    \item Person B: \textit{``Who would like to go out to a restaurant tonight?''} (immigrant event) \\
    Person A: \textit{``I would love to come!''} (triggered event)
\end{enumerate}

In the classic Hawkes model, the message from A in both the first and the second scenario increases the intensity of Person C according to $K_{A \to C}$. This implies that Person C has the same rate of replies to the messages \textit{``Who would like to go out to a restaurant tonight?''} and \textit{``I would love to come!''}, as both messages come from Person A.  However it is clear that the context of A's message is very different in these two scenarios. In the first scenario, A is sending an invite so it is very likely to trigger a response from C. However in the second scenario, it is B sending the invite and A is merely responding to B. Here, C is less likely to reply to A (and more likely to reply to B). The classic Hawkes process has no mechanism to take this context difference into account since it treats all messages sent by A -- and their impact on C -- as identical.

Our Ancestor Hawkes model alleviates this problem. It allows for different reply propensities for messages that start a new messaging cascade (immigrant events) and for those that are a contribution/reply to an ongoing conversation (triggered event). Here, $\Kbold$ governs the influence of immigrant events and $\Lbold$ describes the influence of triggered events. Hence, \textit{``Who would like to go out to a restaurant tonight?''} increases the intensity of Person C by $K_{A \to C}$, since this is an immigrant message. In contrast, the message \textit{``I would love to come!''} increases the intensity of Person C according to $L_{A \to C}$, as it is a triggered message. This allows the rate of reply for $A\to C$ to differ between immigrant ($K_{A \to C}$) and triggered events ($L_{A \to C}$). We show below in Section~\ref{sec_groupchat_app} that through this the Ancestor Hawkes model can capture dynamics in the group chat data that the classic Hawkes is not able to express.


Crucially,  the Ancestor Hawkes model uncovers messaging dynamics that are too difficult to express in the classic Hawkes process model. In a group chat setting, it reveals the individual response preferences of each participant. On one hand, this is interesting from a sociological point of view as it allows for the analysis of hidden community structures similar to \citet{miscouridou_modelling_2018}, who demonstrate this for one-to-one email conversations. On the other hand, our approach could be employed by app developers to infer and summarise reply propensities in a privacy-conscious manner: only the time and sender of messages — not their content — are used. The same distinction between spontaneous starts and replies also arises in other multivariate streams (e.g., contact-centre queues, organisational email without reliable reply metadata, and re-share cascades), where separating exogenous load from reply-driven coupling is likewise useful.

The remainder of the paper is organised as follows. Section~\ref{sec_models_anc} defines the Ancestor Hawkes process and gives stability conditions. Section~\ref{sec_estimation_anc} develops a Bayesian estimation procedure based on a latent-variable Gibbs sampler. Section~\ref{sec_simulated_data} uses simulated data to compare the proposed model with the classic Hawkes process. Section~\ref{sec_groupchat_app} applies the Ancestor Hawkes model to real group chat data, while Section~\ref{sec:timevarying} presents a further robustness study based on data simulated from the fitted group chat model. We conclude with a discussion.

\section{Models} \label{sec_models_anc}

\subsection{Point Process Background}

We work with an $M$-variate point process, represented as a time-ordered sequence of events
$\{(t_i,d_i)\}_{i=1}^N$ on $[0,T]$, where $0<t_1<\cdots<t_N\le T$ are event times and
$d_i\in\{1,\dots,M\}$ indicates the dimension of event $i$.
For each dimension $m$, define the counting process
\[
N_m(t) \;:=\; \sum_{i=1}^N \mathbf{1}\{t_i\le t,\ d_i=m\}, \qquad t\in[0,T],
\]
and collect these into the vector $N(t)=(N_1(t),\dots,N_M(t))$.
Let $\mathcal H_t$ denote the observable history up to time $t$, i.e., all event times and their dimensions
$\{(t_i,d_i): t_i<t\}$. A (conditional) intensity for dimension $m$ is a nonnegative process
$\lambda_m(t)$ such that for small $\Delta t>0$,
\[
\Pr\{N_m(t+\Delta t)-N_m(t)=1 \mid \mathcal H_t\}
= \lambda_m(t)\,\Delta t + o(\Delta t),\qquad
\Pr\{N_m(t+\Delta t)-N_m(t)\ge 2 \mid \mathcal H_t\}=o(\Delta t).
\]
In the multivariate case we also assume that simultaneous jumps across distinct dimensions occur
with negligible probability:
\[
\Pr\{N_m(t+\Delta t)-N_m(t)=1,\; N_{m'}(t+\Delta t)-N_{m'}(t)=1 \mid \mathcal H_t\}
= o(\Delta t)\quad \text{for } m\neq m'.
\]

\subsection{Classic Hawkes Process}

The Hawkes process is used to capture clustering/self-exciting behavior in point processes. The standard multivariate Hawkes process has the following intensity function, given in Equation \ref{eqn:intro} and restated here for convenience:
$$
 \lambda_m(t)  = \mu_m(t)+ \sum_{i: t_i < t} K_{d_i \to m} \, g_{d_i \to m} (t-t_i),   
$$
where we have suppressed the dependence on $\mathcal{H}_t$ for notational convenience. The three components of this intensity are:

\begin{enumerate}
    \item The background rates $\mu_m(t) >0$ determine the rate at which immigrant events are produced. Often they will be a function of time in order to capture underlying seasonality or trends that are not driven by self-exciting behavior. The choice of background rate is flexible and often application-specific. For example, \citet{mohler2013modeling} uses a Log-Gaussian Cox process, \citet{molkenthin_gp-etas_2022} use a Gaussian Process to represent the background rate, and \citet{markwick_bayesian_2020} employs a Dirichlet process.
    \item The influence magnitudes $K_{d_i \to m}\geq 0$  control the influence that an event in dimension $d_i$ has on dimension $m$. Since these magnitudes are required to be non-negative, the Hawkes process is referred to as `self-exciting' -- the occurrence of an event causes the process magnitude to increase, which makes it more likely that further events will happen soon after. This naturally leads to events occurring in clusters/bursts.
    \item The influence kernels  $g_{d_i \to m}(\cdot)$  determine how the influence is `spread out' over time. Typically these are required to be monotonic decreasing, so that the impact of each event on the process decays over time. To aid interpretation, we require that $\int_0^{\infty} g(z) dz =  1$ for each $d \to m$ kernel, in which case the $K_{d_i \to m}\geq 0$  parameters can be interpreted as the average number of additional events in dimension $m$ that are triggered by an event in dimension $d_i$. A common choice for the kernels is the Exponential distribution $g(z) = \beta e^{-\beta z}$ (where the $\beta$'s might vary for each $d \to m$ pair), although other kernels, such as heavy-tailed power laws, are commonly used in seismology \citep{ogata_statistical_1988,kolev_semiparametric_2020}.
\end{enumerate}



\subsection{Branching Process Interpretation}

A linear multivariate Hawkes process admits a cluster representation \citep{hawkes_cluster_1974}. Here, immigrant events arrive in each dimension according to the background rate $\mu_m(t)$ and start new clusters of triggered events. Each event (immigrant or triggered) generates direct offspring in any dimension via inhomogeneous Poisson processes with rates given by the influence parameters and temporal kernels (as in the Hawkes intensity). Iterating this yields generations of events; the observed data are then the union of all cluster trees.

Based on this, it is common  \citep{rasmussen_bayesian_2013,ross_bayesian_2021} to introduce latent variables which show the direct parent of each event $i$:

\[
B_i \;=\;
\begin{cases}
0, & \text{if event $i$ is an immigrant (generated by the background rate)}\\
k, & \text{if event $i$ was directly triggered by event $k$}
\end{cases}
\]
We  collect these as $\mathbf{B}=\{B_1,\dots,B_N\}$.  These variables are latent since we do not directly observe whether events are immigrants or triggered. We also define
\begin{align}
    S_{i,m} = & \{j: B_j = i \text{ and } d_j = m\}  \text{ for } i = 0 \dots N, m = 1 \dots M ,
\end{align}

so that $S_{i,m}$ contains the indices of \emph{direct children} of event $i$ that occur in dimension $m$. 
The special case $S_{0,m}$ collects the \emph{immigrant} events in dimension $m$. 

In the classic multivariate Hawkes model above, each event $i$ in dimension $d_i$ produces potential offspring on target dimension $m$ with rate $K_{d_i \to m}\, g_{d_i \to m}(t-t_i)$; our Ancestor variants will extend this by allowing these rates to vary depending on whether event $i$  is an immigrant or triggered.

\subsection{Ancestor Hawkes}

We now present our new Ancestor Hawkes model. Here, immigrant and triggered events differ in their influences. Hence, the branching structure is directly incorporated into the parameter architecture. In this case, the influence of an event $y_i$ onto dimension $m$ does not only depend on the dimension of $y_i$ but additionally on its (latent) branching variable $B_i$. When $B_i = 0$  and hence $y_i$ is an immigrant event, then it influences other dimensions according to the influence magnitude matrix $\Kbold$.  Alternatively, if $y_i$ is a triggered event and therefore $B_i > 0$ then the influence is governed by a different magnitude matrix $\Lbold$. For the Ancestor Hawkes model, we propose the intensity function
\begin{equation}
    \lambda_m(t) = \mu_m(t) + \underbrace{\sum_{\substack{i: t_i < t, \\B_i = 0}} K_{d_i \to m} \, g_{d_i \to m} (t-t_i)}_{\text{contributions from immigrant events}} + \underbrace{\sum_{\substack{i: t_i < t, \\B_i > 0}} L_{d_i \to m} \, h_{d_i \to m} (t-t_i)}_{\text{contributions from triggered events}}  \label{eq_intensity_ancestor} ,
\end{equation}
 This leads to parameters $\thetabold = (\mubold, \Kbold, \Lbold, \thetabold_g, \thetabold_h)$, where both $\Kbold$ and $\Lbold$ are matrices of dimension $M \times M$. Hence, there are $2M^2$ influence magnitude parameters in this model. In addition,  $\mubold$ denotes the parameters of the background rate,  $\thetabold_g$ parameterises  $g_{m_1 \to m_2}(\cdot)$  (the influence kernel for immigrant events), and $\thetabold_h$ contains the parameters for  $h_{m_1 \to m_2}(\cdot)$ (the influence kernel for triggered events). When referring to the parameters that govern the influence from dimension $m_1$ onto $m_2$ we write $\theta_{h \mid m_1 \to m_2}$. As above, we will always choose influence kernels $g$ and $h$ that integrate to 1, so that the entries of $\Kbold$ and $\Lbold$ can be interpreted as expected direct offspring counts. 

The intensity in Equation \ref{eq_intensity_ancestor} is conditional on the enlarged history containing both
the observed events and their latent branching variables. In real-data applications
these branching variables are unobserved, so inference is carried out by sampling them
jointly with the model parameters. While simulation from the process is straightforward
using the cluster representation above, posterior inference is more complicated; in
Section~\ref{sec_estimation_anc} we derive a Gibbs sampler which samples the latent parent variables and
draws from the posterior distribution of the model parameters.

\paragraph{Restricted variant.}

One possible restricted version of the Ancestor Hawkes model is obtained by setting
\(L_{j\to m}=0\) for \(j\neq m\), so that triggered events may only generate further
events in their own dimension. This prevents triggered cascades from propagating
across dimensions and may be useful when a more conservative interpretation of
cross-dimensional influence is desired. We focus here on the unrestricted Ancestor
Hawkes model.




\subsection{Stability of the Process}
We now discuss stability and stationarity for the Ancestor Hawkes process. Informally, a Hawkes process is \emph{stable} (subcritical) if the branching process associated with each immigrant event has finite expected total size. Under standard regularity conditions on the kernels and the background rate, this rules out explosive cascades in finite time. If the process parameters are constant, then this also guarantees the existence of a \emph{stationary} version of the process.

\textbf{Univariate intuition}: for a univariate Hawkes with $\lambda(t)=\mu(t)+K\sum_{i:t_i<t} g(t-t_i)$ and $\int_0^\infty g(u)\,du=1$, each event produces on average $K$ direct offspring. If $K<1$ then the effect of any event dies out in expectation; if $K\ge 1$ the associated branching process is not subcritical and a finite-mean stationary version does not exist. As such, stability requires $K < 1$.

\textbf{Multivariate standard Hawkes}: for the $M$-variate case (Equation~\ref{eqn:intro}), stability holds when the spectral radius of the reproduction matrix is less than one, $\rho(\mathbf{K})<1$; with constant $\mu$ this yields a unique stationary solution \citep[see, e.g.,][]{deutsch2025cannibalisation,bacry_hawkes_2015}.

\textbf{Ancestor Hawkes}: separate (i) immigrant$\to$first-generation influence via $\mathbf{K}$ from (ii) triggered$\to$triggered influence via $\mathbf{L}$. Long-run feedback arises only through the triggered pathway, so stability is governed by $\mathbf{L}$:

\begin{theorem}[Mean behaviour and subcriticality]
Let $g_{j\to d},h_{j\to d}\ge 0$ with $\int_0^\infty g_{j\to d}(u)\,du=\int_0^\infty h_{j\to d}(u)\,du=1$ for all ordered pairs $(j,d)$. Let $\mathbf{K},\mathbf{L}\in\mathbb{R}^{M\times M}_{+}$ denote the immigrant$\to$triggered and triggered$\to$triggered influence matrices, and suppose the background rates $\mu_d(t)$ are locally integrable. Let $r_d(t)$ denote the mean intensity of triggered events in dimension $d$, so that the mean total intensity is
\[
m_d(t):=\mathbb{E}[\lambda_d(t)] = \mu_d(t)+r_d(t).
\]
If $\rho(\mathbf{L})<1$, the triggered branching process is subcritical and the mean triggered intensities satisfy, for each $d=1,\dots,M$,
\[
r_d(t)
=
\sum_{j=1}^M K_{j\to d}\!\int_{0}^{\infty} g_{j\to d}(u)\,\mu_j(t-u)\,du
+
\sum_{j=1}^M L_{j\to d}\!\int_{0}^{\infty} h_{j\to d}(u)\,r_j(t-u)\,du.
\]
Consequently, the mean total intensities are $m_d(t)=\mu_d(t)+r_d(t)$.
\end{theorem}

\begin{proof}[Proof sketch]
Immigrant events on source dimension $j$ arrive with rate $\mu_j(t)$ and seed a first generation process on $d$ with mean contribution
\[
K_{j\to d}\int_0^\infty g_{j\to d}(u)\,\mu_j(t-u)\,du.
\]
Triggered events then reproduce linearly via $\mathbf{L}$ and kernels $h_{j\to d}$, contributing
\[
L_{j\to d}\int_0^\infty h_{j\to d}(u)\,r_j(t-u)\,du.
\]
The crucial point is that the recursive term involves the mean intensity of triggered events, $r_j(t)$, rather than the total mean intensity $m_j(t)$, since immigrant events reproduce through $\mathbf{K}$ rather than through $\mathbf{L}$. Since $\rho(\mathbf{L})<1$, the triggered subtree is subcritical, which yields the stated renewal system for $r(t)$ and hence $m(t)$.
\end{proof}

A corollary of this result is that if $\mu_j(t)\equiv \mu_j$ for all $j$, then $r(t)$ and $m(t)=\bar\lambda$ are constant. Using the column-vector convention in which $\mathbf{K}_{d j}=K_{j\to d}$ and $\mathbf{L}_{d j}=L_{j\to d}$, we obtain
\[
r = (I-\mathbf{L})^{-1}\mathbf{K}\,\mu,
\qquad
\bar\lambda \;=\; \mu \;+\; (I-\mathbf{L})^{-1}\mathbf{K}\,\mu,
\]
or equivalently
\[
\bar\lambda_d
=
\mu_d
+
\sum_{j=1}^M
\big[(I-\mathbf{L})^{-1}\mathbf{K}\big]_{d j}\,\mu_j.
\]
This recovers the standard result for multivariate Hawkes processes (e.g., \citealp{hawkes_cluster_1974,achab2017}) with the first-generation triggered events (via $\mathbf{K}$) split from ongoing reproduction (via $\mathbf{L}$). 

\section{Estimation} \label{sec_estimation_anc}

We now develop a Bayesian Gibbs sampler for the Ancestor Hawkes process. Our approach extends the latent-variable sampler of \citet{ross_bayesian_2021} from the univariate classic Hawkes setting, which we first review. The main additional complication, beyond the multivariate extension, is that the process intensity in Equation \ref{eq_intensity_ancestor} depends on the latent branching variables themselves. As a result, the conditional independence structure used in the classic Hawkes sampler no longer holds, and the branching variables must be updated differently.

We assume that priors $\pi(\cdot)$ have been chosen for each parameter of the model (more details on the priors will be given later). The iterative Gibbs procedure uses a starting value $\thetabold^{(0)}$ and at step $k$ of the procedure produces a sample   $\thetabold^{(k)}$ from the posterior distribution.  Each block is updated from its full conditional; we normalize all categorical weights where relevant.

\subsection{Classic Hawkes with Branching Structure}

We begin with the classic multivariate Hawkes process. The sampler is a straightforward multivariate extension of \citet{ross_bayesian_2021}, although, to our knowledge, this extension has not previously been presented in a unified form. Here, the model parameter vector is $\thetabold = (\mubold, \Kbold, \thetabold_g)$ which we augment with the latent branching variables $\Bbold$, 

\subsubsection{Conditional Likelihood}

Using an argument similar to \citet{ross_bayesian_2021}, the likelihood of a classic multivariate Hawkes process conditional on the latent branching structure $\Bbold$ can be written as follows.

\begin{align}
p(Y\mid\mathbf B,\boldsymbol\mu(\cdot),\mathbf K,\theta_g)
=&
\prod_{m=1}^M
\left\{
\left[
\prod_{i\in S_{0,m}} \mu_m(t_i)
\right]
\exp\!\left(-\int_0^T \mu_m(t)\,dt\right)
\right\}
\nonumber\\
&\times
\prod_{\ell=1}^M
\prod_{\substack{p:\ d_p=\ell}}
\prod_{m=1}^M
\left\{
\exp\!\big(-K_{\ell\to m}G_{\ell\to m}(T-t_p)\big)
\prod_{i\in S_{p,m}}
K_{\ell\to m}g_{\ell\to m}(t_i-t_p)
\right\}.
\end{align}

where  $G_{\ell \to m}(z) = \int_0^z g_{\ell \to m}(x) \, dx$.

This follows from the cluster process representation where we write the total intensity at time $t$ as a superposition of independent Poisson processes, corresponding to the background process, and a triggered process for each of the previous events that occurred before $t$ (independence holds conditional on $\Bbold$). The first term above is the likelihood contribution from the background process, and the other terms represent the contributions from the triggered processes.

As argued by \citet{ross_bayesian_2021} the key reason for conditioning on the latent background variables is that it drastically reduces the correlation between the model parameters. Specifically,  conditional on the branching variables, the parameters for the background rate $\mubold$ are independent of $(\Kbold, \thetabold_g)$, which describe the influences. Additionally, in our multivariate context, $(K_{j \to m}, \theta_{g\mid j \to m})$ are conditionally independent of  $(K_{i \to n}, \theta_{g \mid i \to n})$ when $i \neq j$ or $n \neq m$. These conditional independences are exploited in the subsequent sampling procedure. While our estimation scheme uses a Bayesian approach, we note in passing that these conditional independences would also facilitate faster convergence for maximum likelihood estimation in an Expectation-Maximization approach, similar to that developed by \citet{Veen_2008_estimation} for a univariate Hawkes process.

We now present our Gibbs sampling scheme. As usual, each of the model parameters is sampled in turn, from its full conditional posterior.

\subsubsection{Sample $\Bbold^{(k+1)}$}

First, the latent branching structure is sampled, assuming that there is equal prior probability for $B_i = 0 \dots i-1$, by \begin{equation} p(B_i^{(k+1)} = j \mid \mubold^{(k)}, \Kbold^{(k)}, \thetabold_g^{(k)},Y) \propto \begin{cases} \mu_{d_i}^{(k)}(t_i) ,& \text{if } j = 0 \\ K_{d_j \to d_i}^{(k)} \, g^{(k)}_{d_j \to d_i}(t_i - t_j) , & \text{if } j \in \{1, 2 \dots i-1\} \end{cases}. \end{equation} Each $B_i^{(k+1)}$ for $i = 1 \dots N$ is sampled independently. Based on $\Bbold^{(k+1)}$, the construction of $\mathbf{S}^{(k+1)}$ follows.

\subsubsection{Sample $\mubold^{(k+1)}$}

First assume that $\mu_m$ is constant. The conditional posterior density of $\mu_m$ is
\begin{equation}
    \pi(\mu_m \mid  \Bbold, Y) \propto \pi(\mu_m) \times \mu_m ^{\mid S_{0,m}\mid } \, \text{exp}(- \mu_m \, T).
\end{equation}

When the prior $\pi(\mu_m)$ is chosen to be a Gamma distribution with parameters $(a_\mu, b_\mu)$, then the posterior distribution is conjugate with
\begin{equation}
   \pi(\mu_m \mid  \Bbold,Y ) = \text{Gamma}(a_\mu + |S_{0,m}| , b_\mu + T).
\end{equation}
This means each $\mu_m^{(k+1)}$ can be sampled directly from $\text{Gamma}(a_\mu + |S_{0,m}^{(k+1)}| , b_\mu + T)$ for $m = 1 \dots M$. If the prior is not conjugate, then any Metropolis-Hastings scheme can be used to draw a sample instead.

When $\mu_m(t)$ is time-varying, note that conditional on \(\mathbf B\) the immigrant times \(\{t_i: B_i=0,\ d_i=m\}\) are draws from an inhomogeneous Poisson process with rate \(\mu_m(t)\). We can model \(\mu_m(t)\) as (e.g.) piecewise-constant on a partition \(\{I_b\}\) of \([0,T]\): \(\mu_m(t)=\mu_{m,b}\) for \(t\in I_b\) with bin width \(\Delta_b=|I_b|\). Place independent Gamma priors
\[
\mu_{m,b} \sim \mathrm{Gamma}(a_{\mu,b},\, b_{\mu,b}) \quad \text{(shape–rate)}.
\]
Let \(n_{m,b} := |\{i: B_i=0,\ d_i=m,\ t_i\in I_b\}|\). Then the full conditional is conjugate:
\[
\mu_{m,b} \mid \mathbf B \ \sim\ \mathrm{Gamma}\!\big(a_{\mu,b}+n_{m,b},\ b_{\mu,b}+\Delta_b\big).
\]
This is the standard inhomogeneous Poisson update with counts \(n_{m,b}\) and exposure \(\Delta_b\).

\subsubsection{Sample $\Kbold^{(k+1)}$ }

For each ordered pair $(m_1,m_2)$, define the sufficient statistics
\[
C_{m_1\to m_2} \;:=\; \sum_{\substack{p:\ d_p=m_1}} |S_{p,m_2}|,
\qquad
E_{m_1\to m_2} \;:=\; \sum_{\substack{p:\ d_p=m_1}} G_{m_1\to m_2}(T-t_p),
\]
where $G_{m_1\to m_2}(z) = \int_0^z g_{m_1\to m_2}(x)\,dx$.

Assume a Gamma prior on each entry with **shape–rate** parameterisation
\[
K_{m_1\to m_2}\ \sim\ \mathrm{Gamma}\!\big(a^{(0)}_{m_1\to m_2},\, b^{(0)}_{m_1\to m_2}\big),
\quad\text{density } f(x)\propto x^{a-1}e^{-bx}.
\]
Then the full conditional is
\[
K_{m_1\to m_2}\mid \Bbold,\thetabold_g, Y \ \sim\
\mathrm{Gamma}\!\Big(
a^{(0)}_{m_1\to m_2} + C_{m_1\to m_2}\ ,\
b^{(0)}_{m_1\to m_2} + E_{m_1\to m_2}
\Big),
\]
and entries update independently across $(m_1,m_2)$.

\subsubsection{Sample $\thetabold_g^{(k+1)}$ }

For each ordered pair $(m_1,m_2)$, the full conditional (up to normalization) is
\begin{align}
\pi\!\left(\theta_{g\mid m_1\to m_2}\mid \Bbold,\Kbold,Y \right)
&\propto \pi\!\left(\theta_{g\mid m_1\to m_2}\right) \nonumber\\[2pt]
&\quad\times \prod_{\substack{p:\ d_p=m_1}}
\exp\!\Big(-K_{m_1\to m_2}\,G_{m_1\to m_2}(T-t_p;\,\theta_{g\mid m_1\to m_2})\Big) \nonumber\\[2pt]
&\quad\times \prod_{\substack{p:\ d_p=m_1}}\ \prod_{i\in S_{p,m_2}}
g_{m_1\to m_2}\!\big(t_i-t_p;\,\theta_{g\mid m_1\to m_2}\big). \label{eq:thetarg}
\end{align}
We update each $\theta_{g\mid m_1\to m_2}$ independently using a short Metropolis--Hastings
or slice sampling step.

\subsection{Ancestor Hawkes with Branching Structure}

We now turn to the Ancestor Hawkes model, with parameter vector 
$\thetabold = (\mubold, \Kbold, \Lbold, \thetabold_g, \thetabold_h)$,
where $\thetabold_g$ parameterises the immigrant-event kernels 
$g_{m_1 \to m_2}(\cdot)$ and $\thetabold_h$ parameterises the triggered-event kernels 
$h_{m_1 \to m_2}(\cdot)$. Here, $\Kbold$ governs the influence of immigrant events, 
while $\Lbold$ governs the influence of triggered events. As before, we augment the 
model with the latent branching variables $\Bbold$. By the same superposition argument 
used above, the likelihood of $Y$ conditional on $\Bbold$ can be written as:
\begin{align}
p(Y \mid \thetabold, \Bbold)
=&
\prod_{m = 1}^M
\left[
\prod_{i \in S_{0,m}} \mu_m(t_i)
\right]
\exp\left\{-\int_0^T \mu_m(t)\,dt\right\}
\nonumber\\
&\times
\prod_{m = 1}^M
\prod_{j = 1}^M
\prod_{\substack{p: d_p = j, \\ B_p = 0}}
\left[
\exp\{-K_{j \to m} G_{j \to m}(T - t_p)\}
\prod_{i \in S_{p,m}}
K_{j \to m} g_{j \to m}(t_i - t_p)
\right]
\nonumber\\
&\times
\prod_{m = 1}^M
\prod_{l = 1}^M
\prod_{\substack{q: d_q = l, \\ B_q > 0}}
\left[
\exp\{-L_{l \to m} H_{l \to m}(T - t_q)\}
\prod_{r \in S_{q,m}}
L_{l \to m} h_{l \to m}(t_r - t_q)
\right].
\end{align}
Here we write 
$G_{m_1 \to m_2}(z) = \int_0^z g_{m_1 \to m_2}(x) \, dx$ and 
$H_{m_1 \to m_2}(z) = \int_0^z h_{m_1 \to m_2}(x) \, dx$.

The Gibbs sampler is similar to the classic Hawkes case, but with two important modifications.
First, $(\Kbold, \thetabold_g)$ is updated using only offspring of immigrant events, while
$(\Lbold, \thetabold_h)$ is updated using only offspring of triggered events. Second, the
branching variables are no longer conditionally independent. In the classic Hawkes sampler,
each $B_i$ can be updated independently of the other $B_j$ variables, for $j \ne i$.
In the Ancestor model this is no longer true, because assigning an event to the immigrant
or triggered class changes the way that event contributes to future intensities. The branching
variables therefore become mutually dependent and must be sampled accordingly. We now
present our full Gibbs sampling scheme; as above, this could be modified to allow for an
efficient EM-style maximum likelihood estimation in line with \citet{Veen_2008_estimation},
although we do not pursue this here.

\subsubsection{Sample $\Bbold^{(k+1)}$}

In the Ancestor model, the effect of an event on future intensities depends on whether that event is labelled as immigrant or triggered. If $B_j=0$ (immigrant) then its children are generated via $(\mathbf K,g)$; if $B_j>0$ (triggered) then its children are generated via $(\mathbf L,h)$. Thus each $B_j$ affects both the \emph{incoming} likelihood of event $j$ and the \emph{outgoing} likelihood to its direct children. We update the branching variables one at a time in a reverse-time Gibbs sweep \((j=N,N-1,\ldots,1)\). Conditional on the current values of all other branching variables, the direct child sets of event \(j\) are fixed, so the full conditional for \(B_j\) can be computed by combining the likelihood contribution for event \(j\) as a child with the likelihood contribution of event \(j\) as a parent.

\paragraph{Outgoing contribution for event $j$ (as a parent).}
Let $d_j$ denote the source dimension of event $j$, let $\tau_j:=T-t_j$ denote the remaining observation time after event $j$, and define the kernel primitives
$G_{d\to m}(z):=\int_0^z g_{d\to m}(u)\,du$ and $H_{d\to m}(z):=\int_0^z h_{d\to m}(u)\,du$.
Given the current direct children of $j$, the contribution to the likelihood from $j$ acting as a parent is
\[
L_{\text{out}}(j\mid B_j=0)
=\prod_{m=1}^M \exp\!\big(-K_{d_j\to m}\,G_{d_j\to m}(\tau_j)\big)\,
\prod_{i\in S_{j,m}} K_{d_j\to m}\, g_{d_j\to m}(t_i-t_j),
\]
\[
L_{\text{out}}(j\mid B_j>0)
=\prod_{m=1}^M \exp\!\big(-L_{d_j\to m}\,H_{d_j\to m}(\tau_j)\big)\,
\prod_{i\in S_{j,m}} L_{d_j\to m}\, h_{d_j\to m}(t_i-t_j).
\]
The exponential terms are the right–censored compensators from $t_j$ to $T$ for all potential targets $m$.

\paragraph{Incoming contribution for $j$ (as a child).}
For a candidate parent $k\in\{0,1,\ldots,j-1\}$ with $t_k<t_j$, the contribution of the edge $k\!\to\! j$ is
\[
L_{\text{in}}(j\mid B_j=k)=
\begin{cases}
\mu_{d_j}(t_j)\ \text{(time–varying)}\ \text{or}\ \mu_{d_j}\ \text{(constant)}, & k=0,\\[4pt]
K_{d_k\to d_j}\, g_{d_k\to d_j}(t_j-t_k), & k>0 \ \text{and}\ B_k=0,\\[4pt]
L_{d_k\to d_j}\, h_{d_k\to d_j}(t_j-t_k), & k>0 \ \text{and}\ B_k>0.
\end{cases}
\]

\paragraph{Full conditional and backward sweep.}
With a flat prior over admissible parents, the unnormalised weight for $k$ is
\[
w_j(k)\;=\;L_{\text{in}}(j\mid B_j=k)\times L_{\text{out}}(j\mid B_j=k),
\]
normalized over $k\in\{0,1,\ldots,j-1\}$ with $t_k<t_j$. We sample $B_j$ from these normalized weights. After drawing $B_j$, we update the child sets by removing $j$ from its old parent’s $S$ and adding it to the new parent’s $S$. The reverse-time sweep ensures that $S_{j,m}$ is known at the moment of updating $B_j$, so $L_{\text{out}}(j\mid\cdot)$ is computable without look-ahead.

\subsubsection{Sample $\mubold^{(k+1)}$}

The posterior density of $\mu_m$ is the same as in the classic Hawkes case:
\[
\pi(\mu_m \mid \mathbf B, Y)\ \propto\ \pi(\mu_m)\ \mu_m^{\,|S_{0,m}|}\, \exp(-\mu_m\, T).
\]
With a Gamma prior $\mu_m \sim \mathrm{Gamma}(a_\mu,b_\mu)$ (shape–rate),
\[
\mu_m \mid \mathbf B, Y \sim \mathrm{Gamma}\big(a_\mu + |S_{0,m}|,\ b_\mu + T\big),\qquad m=1,\dots,M.
\]

If $\mu(t)$ is time-varying then this can be handled using (e.g) the histogram estimator described above.

\subsubsection{Sample $\mathbf K^{(k+1)}$ and $\mathbf L^{(k+1)}$}
Given $\mathbf B$, the posteriors factorize across ordered pairs. Define the immigrant and triggered parent sets
\[
J_K := \{j:\ B_j=0\}, \qquad J_L := \{j:\ B_j>0\}.
\]
Then for each $(m_1,m_2)$,
\[
K_{m_1\to m_2}\mid \mathbf B,\theta_g, Y \ \sim\ 
\mathrm{Gamma}\!\Big(
a^{(0)}_{m_1\to m_2} + \!\!\sum_{\substack{j\in J_K\\ d_j=m_1}}\! |S_{j,m_2}|\ ,\
b^{(0)}_{m_1\to m_2} + \!\!\sum_{\substack{j\in J_K\\ d_j=m_1}}\! G_{m_1\to m_2}(T-t_j)
\Big),
\]
\[
L_{m_1\to m_2}\mid \mathbf B,\theta_h, Y \ \sim\ 
\mathrm{Gamma}\!\Big(
a^{(0)}_{m_1\to m_2} + \!\!\sum_{\substack{j\in J_L\\ d_j=m_1}}\! |S_{j,m_2}|\ ,\
b^{(0)}_{m_1\to m_2} + \!\!\sum_{\substack{j\in J_L\\ d_j=m_1}}\! H_{m_1\to m_2}(T-t_j)
\Big).
\]


\subsubsection{Sample $\boldsymbol\theta_g^{(k+1)}$ and $\boldsymbol\theta_h^{(k+1)}$}
The kernel-parameter updates are obtained in the same way as above, but with the parent set restricted according to event type. For $\thetabold_g$, we use only immigrant parents; for $\thetabold_h$, we use only triggered parents.

For $\theta_{g\mid m_1\to m_2}$ (restrict to immigrant parents $J_K$):
\begin{align*}
\pi\!\left(\theta_{g\mid m_1\to m_2}\mid \mathbf B,\mathbf K, Y\right)
&\propto \pi\!\left(\theta_{g\mid m_1\to m_2}\right)
\prod_{\substack{j\in J_K:\\ d_j=m_1}}
\exp\!\Big(-K_{m_1\to m_2}\,G_{m_1\to m_2}(T-t_j;\,\theta_{g\mid m_1\to m_2})\Big)\\
&\qquad\qquad\times \prod_{\substack{j\in J_K:\\ d_j=m_1}}\ \prod_{i\in S_{j,m_2}}
g_{m_1\to m_2}\!\big(t_i-t_j;\,\theta_{g\mid m_1\to m_2}\big).
\end{align*}

For $\theta_{h\mid m_1\to m_2}$ (restrict to triggered parents $J_L$):
\begin{align*}
\pi\!\left(\theta_{h\mid m_1\to m_2}\mid \mathbf B,\mathbf L, Y\right)
&\propto \pi\!\left(\theta_{h\mid m_1\to m_2}\right)
\prod_{\substack{j\in J_L:\\ d_j=m_1}}
\exp\!\Big(-L_{m_1\to m_2}\,H_{m_1\to m_2}(T-t_j;\,\theta_{h\mid m_1\to m_2})\Big)\\
&\qquad\qquad\times \prod_{\substack{j\in J_L:\\ d_j=m_1}}\ \prod_{i\in S_{j,m_2}}
h_{m_1\to m_2}\!\big(t_i-t_j;\,\theta_{h\mid m_1\to m_2}\big).
\end{align*}

\section{Simulation Study -  Constant $\mu_m$} \label{sec_simulated_data}

We now examine simulated data in order to study whether the Ancestor Hawkes can accurately estimate model parameters. We also compare it to the classic Hawkes model to explore the additional flexibility that the Ancestor Hawkes model offers. To this end, we simulate data sets from the Ancestor Hawkes model and provide parameter estimates from both the Ancestor Hawkes and classic Hawkes models. We then examine the parameter posterior distributions obtained in both cases.  In this section, we focus on the case where the background rate $\mu_m$ is constant. A second simulation study, using time-varying background rates matched to those estimated from the real data, is presented later in Section \ref{sec:timevarying}.

For now, we simulate data from the Ancestor Hawkes model using Exponential decay kernels,  under the following three different parameter settings:

\begin{itemize}
\item Scenario 1: M (number of dimensions) = 3. Parameters are chosen as $\mu_m = 0.05$ for all $m = 1 \dots M$, $K_{m_1 \to m_2} = 0.6$ for $m1, m2 = 1 \dots M$, $L_{m_1 \to m_2} = 0.3$ where $m_1 = m_2$ and $L_{m_1 \to m_2} = 0.05$  where $m_1 \neq m_2$. Moreover, we set $ \beta_{m_1 \to m_2} = 2$ and $\gamma_{m_1 \to m_2} = 0.5$ (diagonal/off-diagonal elements hence share the same rate within each family in this scenario).

\item Scenario 2: M = 4 with the following parameters: $\mu_1 =0.05, \mu_2=0.07, \mu_3=0.04, \mu_4 = 0.06$ with influence magnitudes:

\[
\mathbf{K}=
\begin{pmatrix}
0.18 & 0.12 & 0.00 & 0.10 \\
0.00 & 0.16 & 0.12 & 0.00 \\
0.10 & 0.00 & 0.17 & 0.12 \\
0.12 & 0.10 & 0.00 & 0.15
\end{pmatrix}
\qquad
\mathbf{L} =
\begin{pmatrix}
0.30 & 0.22 & 0.20 & 0.00 \\
0.20 & 0.28 & 0.00 & 0.18 \\
0.22 & 0.20 & 0.26 & 0.00 \\
0.00 & 0.22 & 0.20 & 0.30
\end{pmatrix}
\]
and influence kernel rate matrices:

\[
\boldsymbol{\beta} =
\begin{pmatrix}
4 & 3 & 3 & 3 \\
3 & 4 & 3 & 3 \\
3 & 3 & 4 & 3 \\
3 & 3 & 3 & 4
\end{pmatrix}
\qquad
\boldsymbol{\gamma}:=
\begin{pmatrix}
0.8 & 0.5 & 0.5 & 0.5 \\
0.5 & 0.8 & 0.5 & 0.5 \\
0.5 & 0.5 & 0.8 & 0.5 \\
0.5 & 0.5 & 0.5 & 0.8
\end{pmatrix}
\]

This represents a more challenging problem than Scenario 1 because the background rates differ across dimensions,  $\mathbf{K}$ and $\mathbf{L}$  vary entry-wise with some sparsity, and the kernel rates are heterogeneous.

\item Scenario 3: A potential concern with the Ancestor Hawkes model is that identifiability may come mainly from differences in the influence kernel shapes rather than from differences in the influence magnitudes themselves. To test this, we reuse the $\mu_m$, $K$ and $L$ matrices from Scenario 2, but now set all kernel rates equal by taking $\beta = \gamma = 2.4$ for every entry.

\end{itemize}

For each set of parameter values, we simulate a point process with $5000$ events (hence the horizon  $T$ is random). For simplicity, when estimating both the standard and Ancestor Hawkes models we restrict the fitted kernel families to two rate parameters per model -- a shared diagonal rate and a shared off-diagonal rate. Specifically, in both the Ancestor and standard Hawkes models we assume

\begin{align}
    g_{j \to m}(x) & = \beta_{\text{diag}}\, \text{exp} \left( -\beta_{\text{diag}} \, x \right), & \text{if } j = m, \\
    g_{j \to m}(x) & = \beta_{\text{off}}\, \text{exp} \left( -\beta_{\text{off}} \, x \right) ,& \text{if } j \neq m.
    \label{eqn:restriction}
\end{align}
 and in the Ancestor Hawkes model we also have:

\begin{align}
     h_{j \to m}(x) & = \gamma_{\text{diag}}\, \text{exp} \left( -\gamma_{\text{diag}} \, x \right), & \text{if } j = m, \\
    h_{j \to m}(x) & = \gamma_{\text{off}}\, \text{exp} \left( -\gamma_{\text{off}} \, x \right), & \text{if } j \neq m.
    \label{eqn:restriction2}
\end{align}

The influence magnitude matrices $\mathbf{K}$ and $\mathbf{L}$ are unconstrained and we estimate the full matrices. 

We use the following priors in both models, with the priors on the influence magnitudes intended to shrink slightly towards 0 in order to handle cases where sparsity may be present:

\begin{align*}
    \mu_m & \sim \text{ Gamma}(1, 1),   & \text{for } m = 1 \dots M, \\
    K_{m_1 \to  m_2} & \sim \text{ Gamma}(1, 10), & \text{for } m_1 = 1 \dots M, m_2 = 1 \dots M, \\
    \beta_{\text{diag}} & \sim \text{ Gamma}(2, 1), & \\
    \beta_{\text{off}} & \sim \text{ Gamma}(2, 1). & 
\end{align*}

The Ancestor Hawkes has the additional  priors:
\begin{align*}
    L_{m_1 \to m_2} & \sim \text{ Gamma}(1, 10),  & \text{for } m_1 = 1 \dots M, m_2 = 1 \dots M, \\
     \gamma_{\text{diag}} & \sim \text{ Gamma}(2, 1), & \\
    \gamma_{\text{off}} & \sim \text{ Gamma}(2, 1). & 
\end{align*}

For all simulation experiments we ran the Gibbs sampler for 20,000 iterations and discarded the first 5,000 as burn-in. The branching variables, background-rate parameters, and influence magnitudes were updated from their full conditional distributions. The kernel-rate parameters were updated using slice sampling.

\subsection{Results}

\subsubsection{Scenario 1}

For each  scenario, we simulated 200 data sets from the model, using the specified parameters. For each simulated data set, the Gibbs sampler was used to draw 20,000 values from the parameter posterior distributions in both the Ancestor and standard Hawkes models, with the first 5,000 samples discarded as burn-in. Figure \ref{fig:experiment1} shows the posterior distribution for the background rates and the kernel shape parameters. It can be seen that under the Ancestor Hawkes model, the posterior distributions correctly concentrate around the true values, while the classic Hawkes model estimates a higher background rate in all dimensions, showing a structural bias.

\begin{figure}[t]
  \centering

  \begin{subfigure}{\linewidth}
    \centering
    \includegraphics[width=0.8\linewidth]{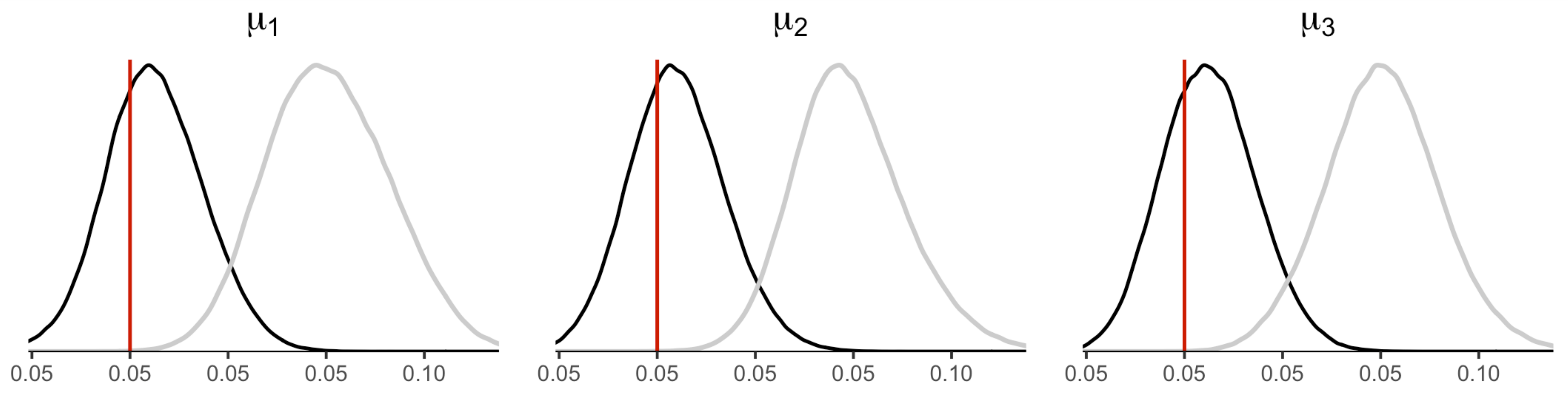}
    \caption{Posteriors for $\mu_1,\mu_2,\mu_3$. The black line shows the posterior distribution in the Ancestor Hawkes model, the grey line shows the posterior distribution in the classic Hawkes model. The vertical red lines indicate the true parameter values.}
  \end{subfigure}

  \vspace{0.6em}

  \begin{subfigure}{\linewidth}
    \centering
    \includegraphics[width=1.1\linewidth]{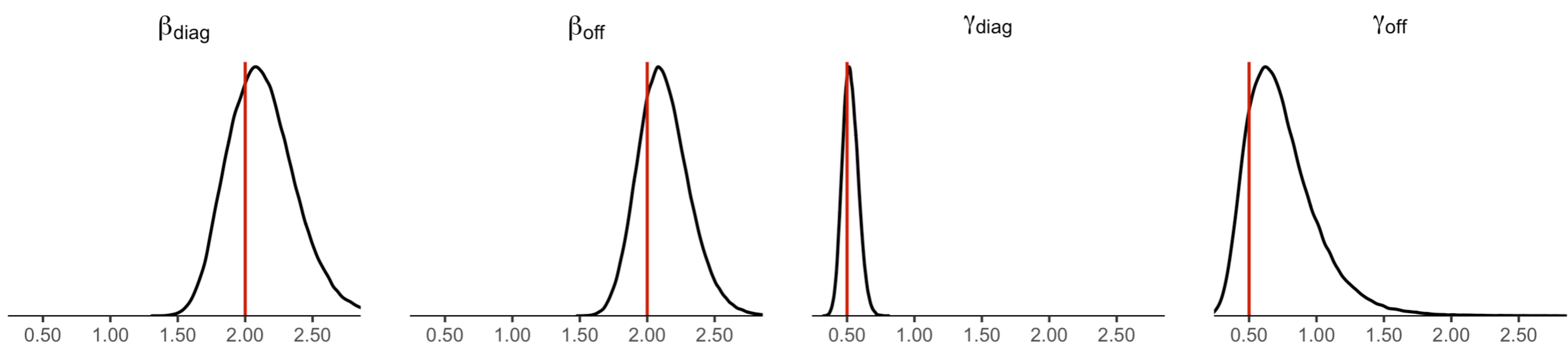}
    \caption{Posteriors for $\beta_{\mathrm{diag}}, \beta_{\mathrm{off}}, \gamma_{\mathrm{diag}}, \gamma_{\mathrm{off}}$ in the Ancestor Hawkes model. The vertical red lines indicate the true parameter values.}
  \end{subfigure}

  \caption{Scenario 1 posterior distributions for $\mu_m$ and the kernel rates.}
  \label{fig:experiment1}
\end{figure}

Next, Figure~\ref{fig:experiment1K} shows the posterior distributions for all influence magnitude parameters. In the classic Hawkes model, this is just the single $\Kbold$ matrix, while for the Ancestor Hawkes both $\Kbold$ and $\Lbold$ are estimated. Here it is evident that the posterior distributions for the classic Hawkes model attempt to find a middle ground between the distributions of $K$ and $L$, essentially blurring the two together. In contrast, the posteriors under the Ancestor model correctly concentrate around the true parameter values. The posteriors for the $K$ parameters exhibit a mild left bias, which is consistent with the shrinkage effect of the $\text{Gamma}(1,10)$ prior.

\begin{figure}[t]
    \centering
    \includegraphics[width=.8\linewidth]{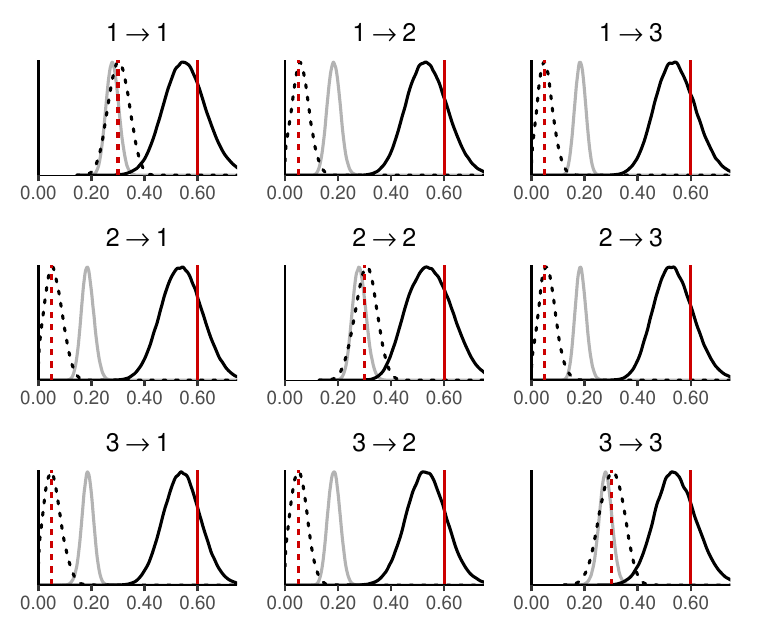}
    \caption{Scenario 1: Posterior distribution estimates for the influence magnitudes. The black solid and dashed lines respectively show the posterior distribution for $\Kbold$ and $\Lbold$ in the Ancestor Hawkes model, while the red vertical lines show the true parameter values. The grey lines show the corresponding $\Kbold$ posterior distribution in the classic Hawkes model -- note that the classic Hawkes seems to interpolate between $\Kbold$ and $\Lbold$. } 
    \label{fig:experiment1K}
\end{figure}

\subsubsection{Scenarios 2 and 3}

Scenarios 2 and 3 are intended to demonstrate that the Ancestor Hawkes model correctly recovers the model parameters in a more complicated situation where the background rates and influence magnitudes vary across dimensions. The simulation strategy is the same as above, using 200 different data sets simulated from the scenario parameters, with 20,000 Gibbs samples taken for each data set. Figure \ref{fig_KLK_comparison_sim}  shows the averaged posterior distribution for $\Kbold$ and $\Lbold$ in both scenarios. It can be seen that these posteriors correctly concentrate around the true values, as expected. Notably, the posterior distributions in both scenarios are extremely similar, demonstrating that the correct identification of the influence magnitudes does not depend on having a clear separation between the influence kernel shapes, since in Scenario 3 the immigrant and triggered kernel rate matrices are identical.

\begin{figure}[!htp]
    \centering
    \includegraphics[width= 1.05\linewidth]{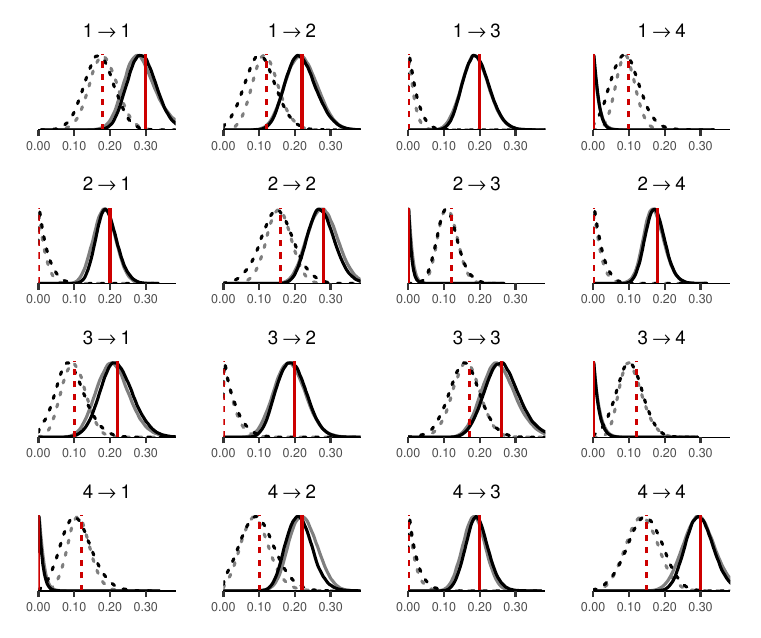}
    \caption{Scenarios 2 and 3: Average posterior distributions for the Ancestor Hawkes influence magnitudes. The solid lines show the posterior distribution for $\Kbold$ and the dashed lines show the posterior distribution for $\Lbold$. The black lines correspond to Scenario 2 while the grey lines (which mostly overlap with the black) correspond to Scenario 3. The vertical red lines indicate the true values (which are solid for $\Kbold$ and dashed for $\Lbold$).} 
    \label{fig_KLK_comparison_sim}
\end{figure}

\section{Group Chat Data Application} \label{sec_groupchat_app}

To highlight the benefits of the Ancestor Hawkes model we have collected data from a group chat setting. The data was downloaded from Facebook Messenger, where the chat was hosted by one of the authors of this study. All chat participants have given consent for their anonymized data to be analyzed and the author obtained the required ethical approval through their host institution. The anonymised timestamp/sender data used in this analysis are publicly available on GitHub at \url{https://github.com/isabella-de/AncestorHawkes}.

The data comes from a chat group with $9$ participants. Here, each message sent to the group chat is simultaneously received by all members and everyone has equal opportunity to participate. This setting is substantially different from, for example, the email data set examined by \citet{miscouridou_modelling_2018}, where each email has exactly one receiver. 

The full data set spans $10,705$ messages over approximately $3.5$ years (from $11^{th}$ of October $2019$ to $29^{th}$ of April 2023). The group chat was started with a subset of participants and additional members were added over time, but for most of the examined period, all $9$ participants were part of the group chat. The data set only contains the sender and times at which a message was sent. No additional information (e.g.\ content of the message, whether it was a reply) is available. Figure~\ref{fig_no_messages} displays the total number of messages sent by each person. Figure~\ref{fig_observations} plots a subset of messages during a one-year period. 

\begin{figure}[!htp]
    \centering
    \includegraphics[width=.99\linewidth]{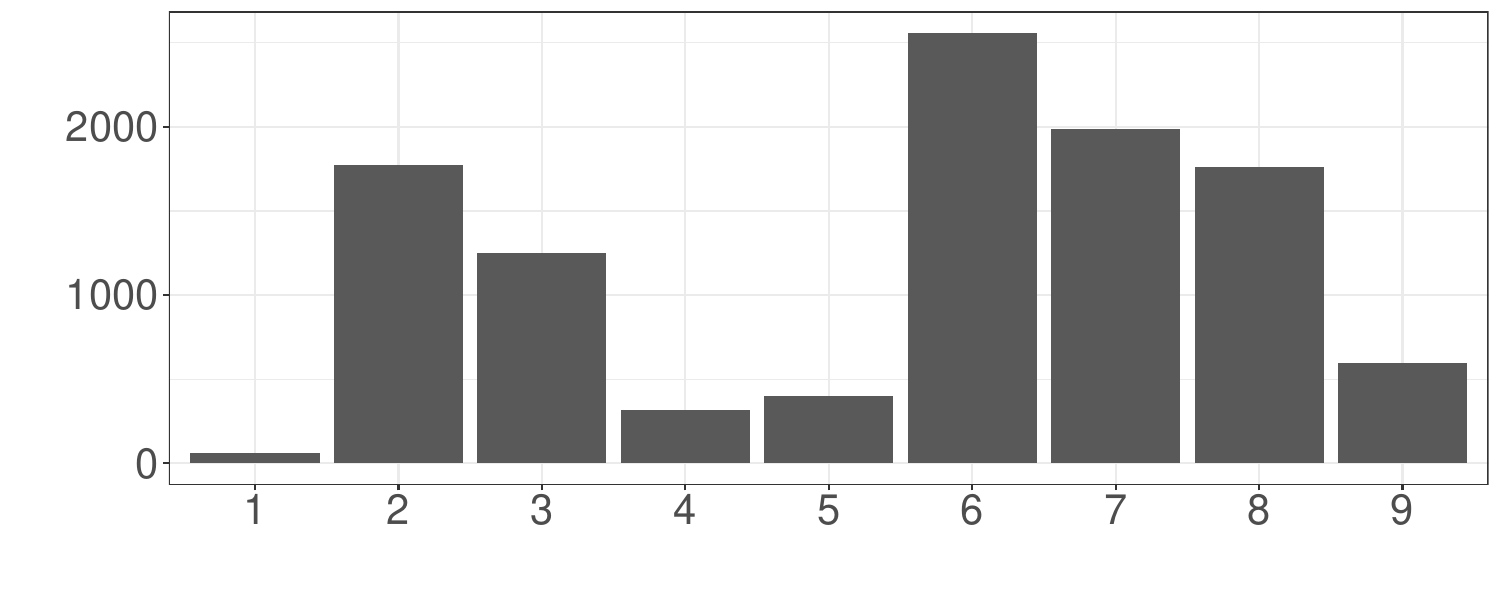}
    \caption{Total number of messages sent by each participant in the group chat.} 
    \label{fig_no_messages}
\end{figure}

\begin{figure}[!htp]
    \centering
    \includegraphics[width=.99\linewidth]{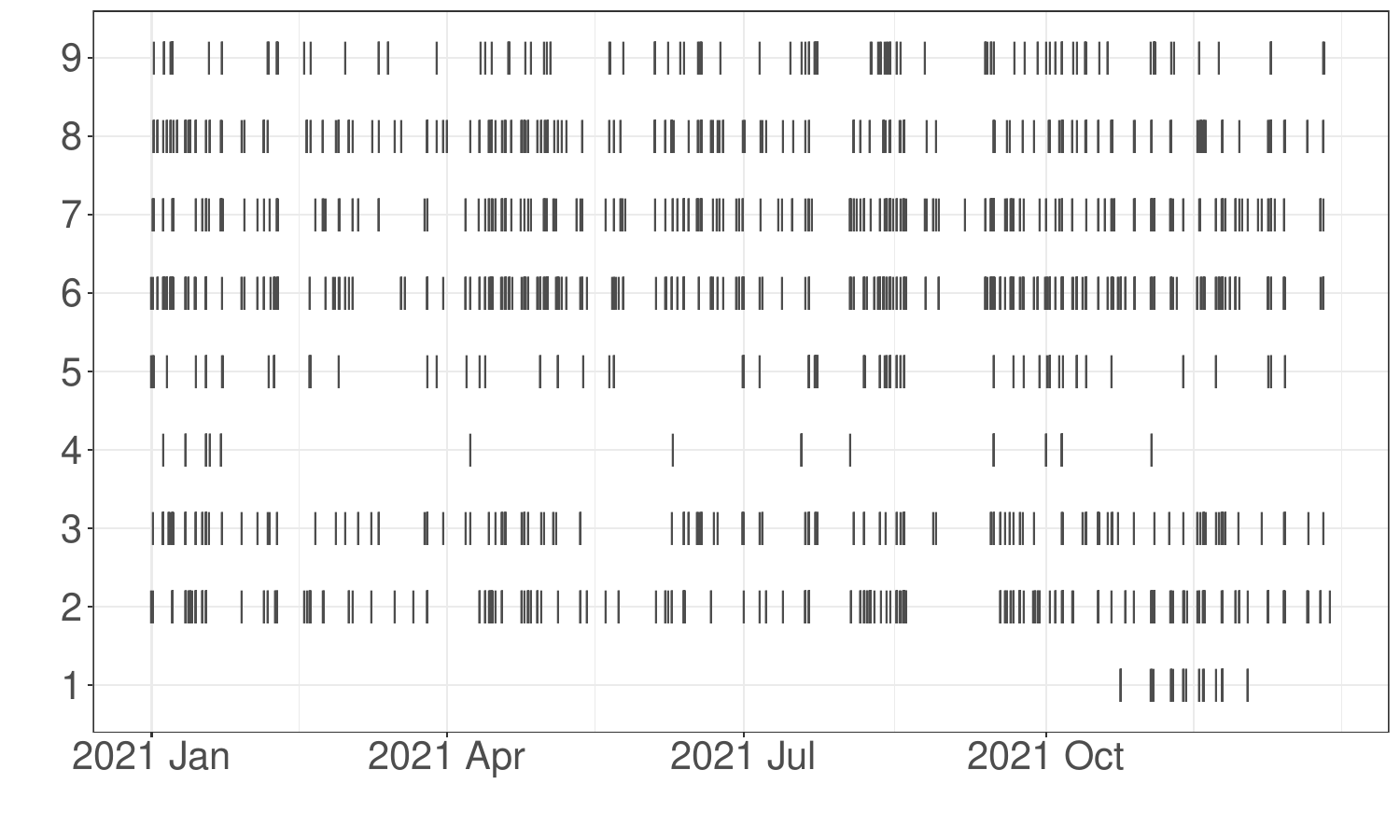}
    \caption{Messages sent by each participant in the group chat in $2021$. All participants had already been added to the chat. Each vertical bar indicates a message sent.} 
    \label{fig_observations}
\end{figure}

First, we examine broad temporal patterns in the data. Figure~\ref{fig_messages_monthly} shows the total number of messages sent per month. This remains somewhat stable over time, apart from a particularly high number of messages in March 2020. This coincides with the start of lockdowns across Europe (where all participants were based at that time) due to the Covid-19 pandemic.

\begin{figure}[!htp]
    \centering
    \includegraphics[width=.99\linewidth]{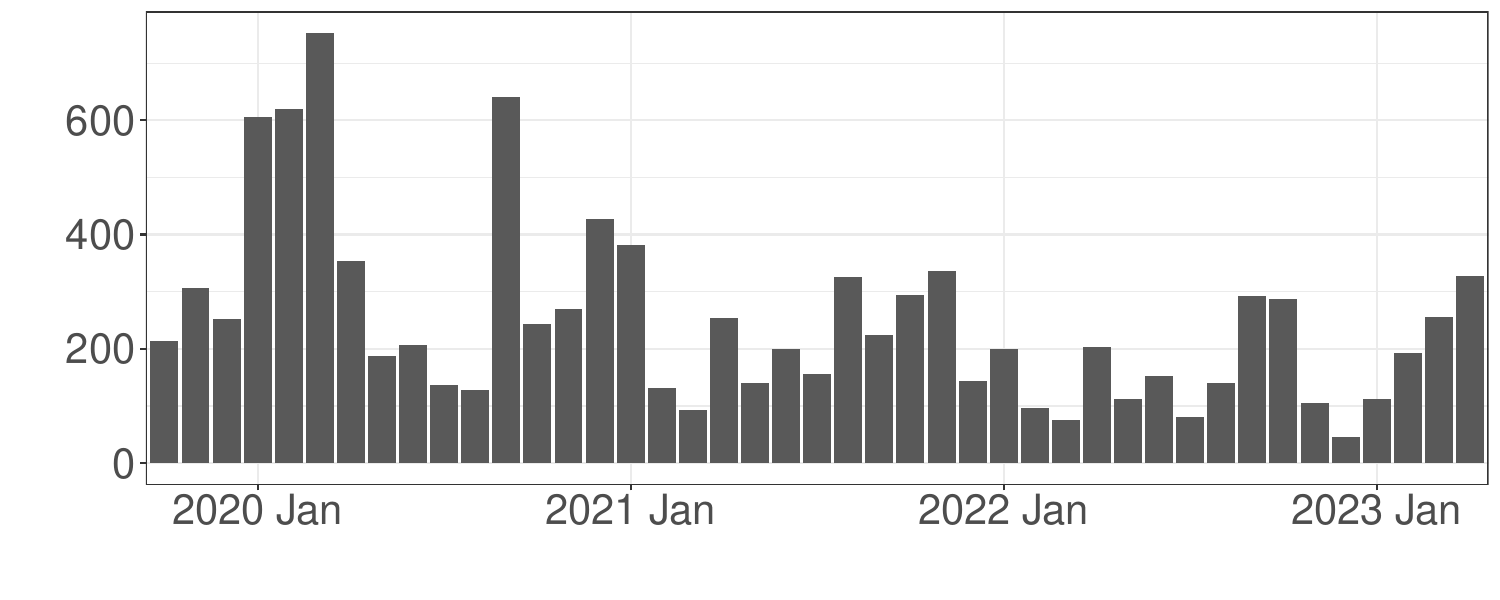}
    \caption{Number of messages sent per month.} 
    \label{fig_messages_monthly}
\end{figure}

Figure~\ref{fig_messages_weekday} shows the number of messages sent for each day of the week. The colour indicates which hour of the day a message was sent. The hours of the day are grouped into four categories each containing a quarter of the day. While there is a big difference within a day between working hours (biggest region in pink on the bottom) and late at night (barely visible orange on the top), these numbers stay similar across days of the week.

\begin{figure}[!htp]
    \centering
    \includegraphics[width=.99\linewidth]{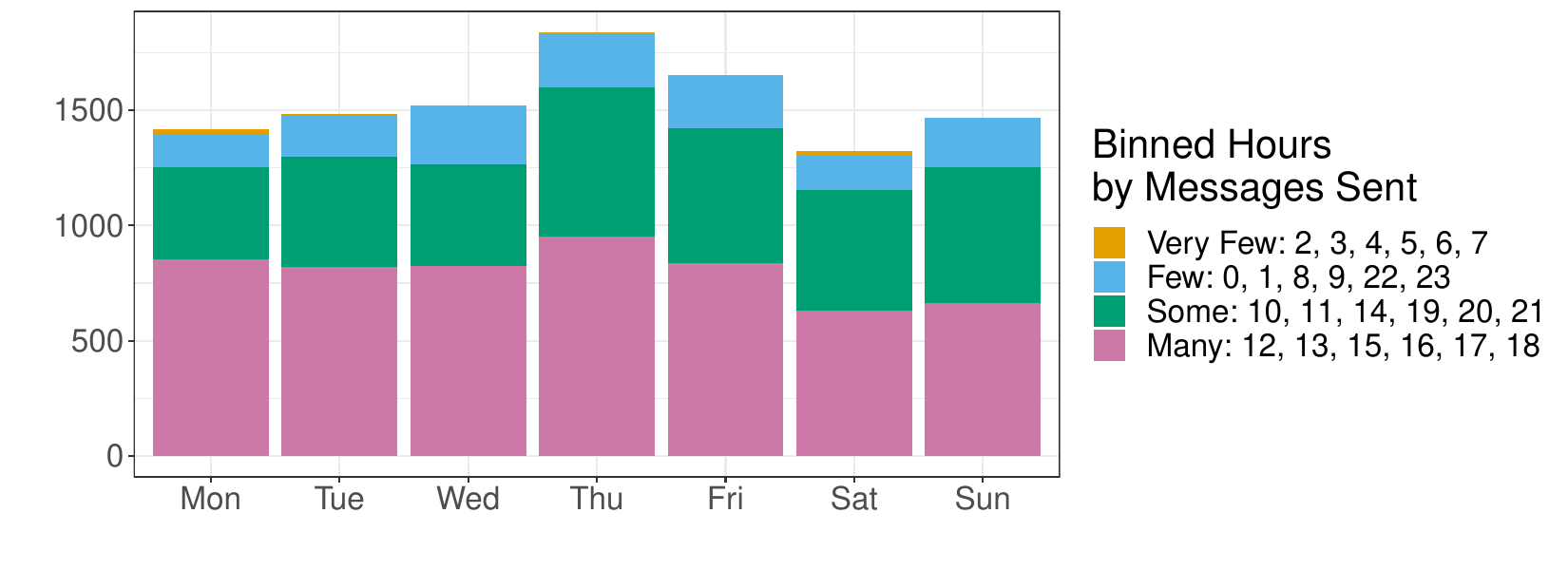}
   \caption{Number of messages sent per weekday. The colour indicates which hour of the day (grouped into four bins as described in the legend) the messages were sent. } 
    \label{fig_messages_weekday}
\end{figure}

For our analysis, we  model the messages that were sent in $2021$ only.

\subsection{Model}

To model the temporal patterns of messages being sent, we utilise an Ancestor Hawkes model where each participant is represented by a dimension and its respective intensity function. As a comparison model, we also fit a classic Hawkes process on the same data.

For both models, a non-constant background rate is used. Hence, the intensity function of the Ancestor Hawkes in dimension $m$ for $m \in \{1,2,\ldots,9\}$ is

\begin{equation}
    \lambda_m(t) = \mu_m(t) + \sum_{\substack{i: t_i < t, \\B_i = 0}} K_{d_i \to m} \, g_{d_i \to m} (t-t_i) + \sum_{\substack{j: t_j < t, \\B_j > 0}} L_{d_j \to m} \, h_{d_j \to m} (t-t_j).
\end{equation}

\subsubsection{Background Rate}

In both the Ancestor and classic Hawkes model,  we specify the  background (immigrant) intensity as a separable, multiplicative seasonal process. For each user $m\in\{1,\dots,9\}$ at time $t\in[0,T_{\text{obs}}]$ we set
\[
\mu_{m}(t)\;=\;\alpha_{m}\;\theta_{\text{hour}}\!\big(h(t)\big)\,
\theta_{\text{wday}}\!\big(w(t)\big)\,
\theta_{\text{month}}\!\big(\mathrm{mon}(t)\big),
\]
where $\alpha_{m}>0$ is a dimension--specific scale and
$\theta_{\text{hour}}:\{1,\dots,24\}\to\mathbb{R}_{\ge0}$,
$\theta_{\text{wday}}:\{1,\dots,7\}\to\mathbb{R}_{\ge0}$,
$\theta_{\text{month}}:\{1,\dots,12\}\to\mathbb{R}_{\ge0}$ are seasonal factors for hour–of–day, day–of–week, and month–of–year. The binning maps are $h(t)\in\{1,\dots,24\}$ (local hour), $w(t)\in\{1,\dots,7\}$ (weekday, with Monday=1), and $\mathrm{mon}(t)\in\{1,\dots,12\}$ (calendar month).

To avoid non–identifiability, each seasonal vector is constrained to have exposure–weighted mean 1 over the observation window. Let
$E(h,\,d,\,m_0)=\int_0^{T_{\text{obs}}}\mathbf{1}\{h(t)=h,\,w(t)=d,\,\mathrm{mon}(t)=m_0\}\,dt$
be the realized exposure (precomputed once as a $24\times7\times12$ tensor), and define weights
$w^{(\text{hour})}_h\propto\sum_{d,m_0}E(h,d,m_0)$,
$w^{(\text{wday})}_d\propto\sum_{h,m_0}E(h,d,m_0)$,
$w^{(\text{month})}_{m_0}\propto\sum_{h,d}E(h,d,m_0)$,
each normalized to sum to one. We enforce
$\sum_h w^{(\text{hour})}_h\,\theta_{\text{hour}}(h)=
\sum_d w^{(\text{wday})}_d\,\theta_{\text{wday}}(d)=
\sum_{m_0} w^{(\text{month})}_{m_0}\,\theta_{\text{month}}(m_0)=1$
and rescale $\alpha$ accordingly, so that $\frac{1}{T_{\text{obs}}}\!\int_0^{T_{\text{obs}}}\mu_m(t)\,dt=\alpha_m$.

Conditional on the current branching structure, the seasonal parameters $(\alpha_m,  \theta_{\text{hour}}, \theta_{\text{wday}}, \theta_{\text{month}})$ are updated inside the Gibbs sampler by applying the inhomogeneous Poisson likelihood to the immigrant-assigned events only, with the exposure-weighted mean-one constraints imposed after each update by rescaling the seasonal factors and compensating through the corresponding \(\alpha_m\).

\subsubsection{Influence Kernel}
As in the simulated data example we limit the number of distinct parameters in the influence kernel (but not in the magnitude matrices $\Kbold$ and $\Lbold$). We assume that the shape is the same for all self-influences, and that a different shape parameter is used for the cross-influences, such that
\begin{align}
    g_{j \to m}(x) & = \beta_{\text{diag}}\, \exp \left( -\beta_{\text{diag}} \, x \right), & \text{if } j = m ,\\
    g_{j \to m}(x) & = \beta_{\text{off}}\,  \exp \left( -\beta_{\text{off}} \, x \right), & \text{if } j \neq m. \nonumber
\end{align}
and:

\begin{align}
     h_{j \to m}(x) & = \gamma_{\text{diag}}\,  \exp \left( -\gamma_{\text{diag}} \, x \right), & \text{if } j = m, \\
    h_{j \to m}(x) & = \gamma_{\text{off}}\,  \exp \left( -\gamma_{\text{off}} \, x \right), & \text{if } j \neq m. \nonumber
\end{align}
Both the $g$ and $h$ kernels are normalised to integrate to one, so that the entries of $\Kbold$ and $\Lbold$ retain their interpretation as expected direct offspring counts. As above, the influence magnitudes $\Kbold$ and $\Lbold$ do not otherwise have any restrictions.

\subsection{Prior Choice}

We complete our model by specifying the parameter priors. We use the same priors on the influence kernel shapes and magnitudes that were previously defined in Section \ref{sec_simulated_data}. The $\alpha_m$ parameters in the seasonal background rate are given the same $\mathrm{Gamma}(1, 1)$ prior that was previously used for $\mu_m$ since they play the same role in the model. We complete the specification by assigning  independent $\mathrm{Gamma}(1,1)$ priors to all the seasonal parameters specifying the hourly, weekly and monthly components of the background rate.

\subsection{Results}

The classic and Ancestor Hawkes models were fitted to the group-chat data. For each model, the Gibbs sampler was used to draw 20,000 values from the parameter posterior distributions, with the first 5,000 samples discarded as burn-in. Representative MCMC trace plots are shown in Appendix A. These plots were used as a basic diagnostic check for stability of the retained chain and show no visible long-term drift in the monitored posterior summaries. In the numerical summaries below, posterior point estimates are posterior means; the figures show the corresponding posterior distributions.

Due to the relatively large number of parameters we only present selected illustrative plots here. First, we investigate the posterior distributions of  influence magnitude parameters. The top row of Figure~\ref{fig_KL_ancestor_only} compares the posterior estimates for $K_{7 \to m}$ and $L_{7 \to m}$, where the special case of $m = 7$ (self-influence) is displayed in a dotted line. These parameters specify the increase in the other participants' intensity functions due to an immigrant message ($K_{ 7\to m}$) or a follow-up message ($L_{7 \to m}$) from the $7^{th}$ participant. It is evident that immigrant messages solicit more responses. This is true across all participants and in line with our understanding of messaging dynamics as an immigrant message may possibly be a question to the whole group (and hence invites replies from everyone) while a triggered/follow-up message is perhaps more likely to be a reply directed at a subset of group members.

The bottom row of Figure~\ref{fig_KL_ancestor_only} depicts the converse parameters $K_{m \to 7}$ and  $L_{m \to 7}$, which control the increase in the message rate of the $7^{th}$ participant in response to an immigrant and triggered message from someone else. For example, this participant is more likely to respond to an immigrant message from participant 3 than to one from participant 5 ($\hat{K}_{3 \to 7} = 0.27$ versus $\hat{K}_{5 \to 7} = 0.05$). In addition, their own messages show strong self-excitation, with $\hat{K}_{7 \to 7} = 0.37$ and  $\hat{L}_{7 \to 7} = 0.22$. Thus an immigrant message from participant 7 generates, on average, $0.37$ direct same-participant triggered messages, while each subsequent triggered message from participant 7 generates, on average, a further $0.22$ direct same-participant triggered messages. Such a comparably large self-influence can be observed for most participants. This corresponds to the intuitive notion that people often send multiple messages about the same topic in rapid succession.

The plots in Figure~\ref{fig_KL_ancestor_only} also show a substantial difference in the behavior of messages in reply chains. For example, we have $\hat{K}_{3 \to 7} = 0.27$ but $\hat{L}_{3 \to 7} = 0.11$. This shows that when user $3$ initiates a new conversation in the group chat (i.e. generate an immigrant message) then user $7$ will send $0.27$ extra messages on average in response. However when the message sent by user $3$ is a reply to an existing conversation (i.e. they produce a triggered event) then user $7$ will only send $0.11$ extra messages on average. This makes intuitive sense and justifies the reasoning between the Ancestor Hawkes model -- an immigrant message which initiates a conservation is likely to induce a fundamentally different reply pattern to a message which is posted in response to an ongoing conversation. In the latter case, the message is likely to be (e.g.) a reply to one specific user and is hence less likely to trigger replies from other users. Similar patterns were observed for all participants in the group chat, not just user $7$.

\begin{figure}
    \centering
    \includegraphics[width=.99\linewidth]{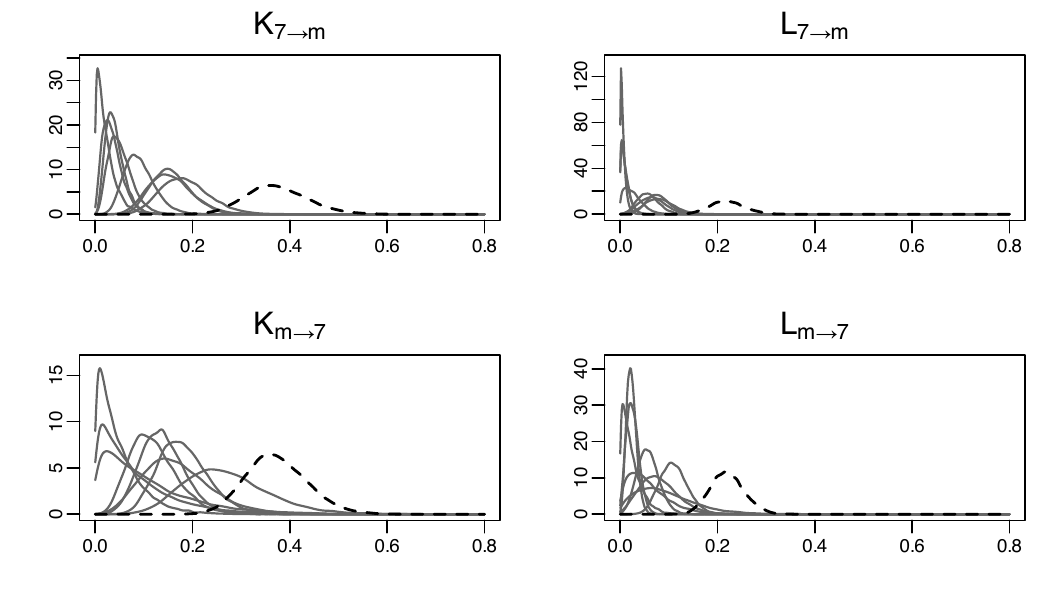}
    \caption{Posterior distributions for the influence magnitude parameters in the Ancestor Hawkes model for the $7^{th}$ participant. The dotted line shows the special case where $m = 7$, i.e.\ self-influence.} 
    \label{fig_KL_ancestor_only}
\end{figure}

Next, Figure \ref{fig_c_scalar} shows the posterior distribution for the $\alpha$ parameter in the background rate, rescaled to correspond to the average number of immigrant messages a day.  Interestingly, Participant~$6$ and $7$ have rather similar values (with posterior means of 0.33 and 0.32 respectively), however participant~$6$ has sent a total of $735$ messages in the analysed data set, whereas Participant~$7$ has only sent $368$ messages. This indicates that Participant~$6$ is more likely to respond to messages, even though both initiate conversations at similar rates.

\begin{figure}
    \centering
    \includegraphics[width=0.7\linewidth]{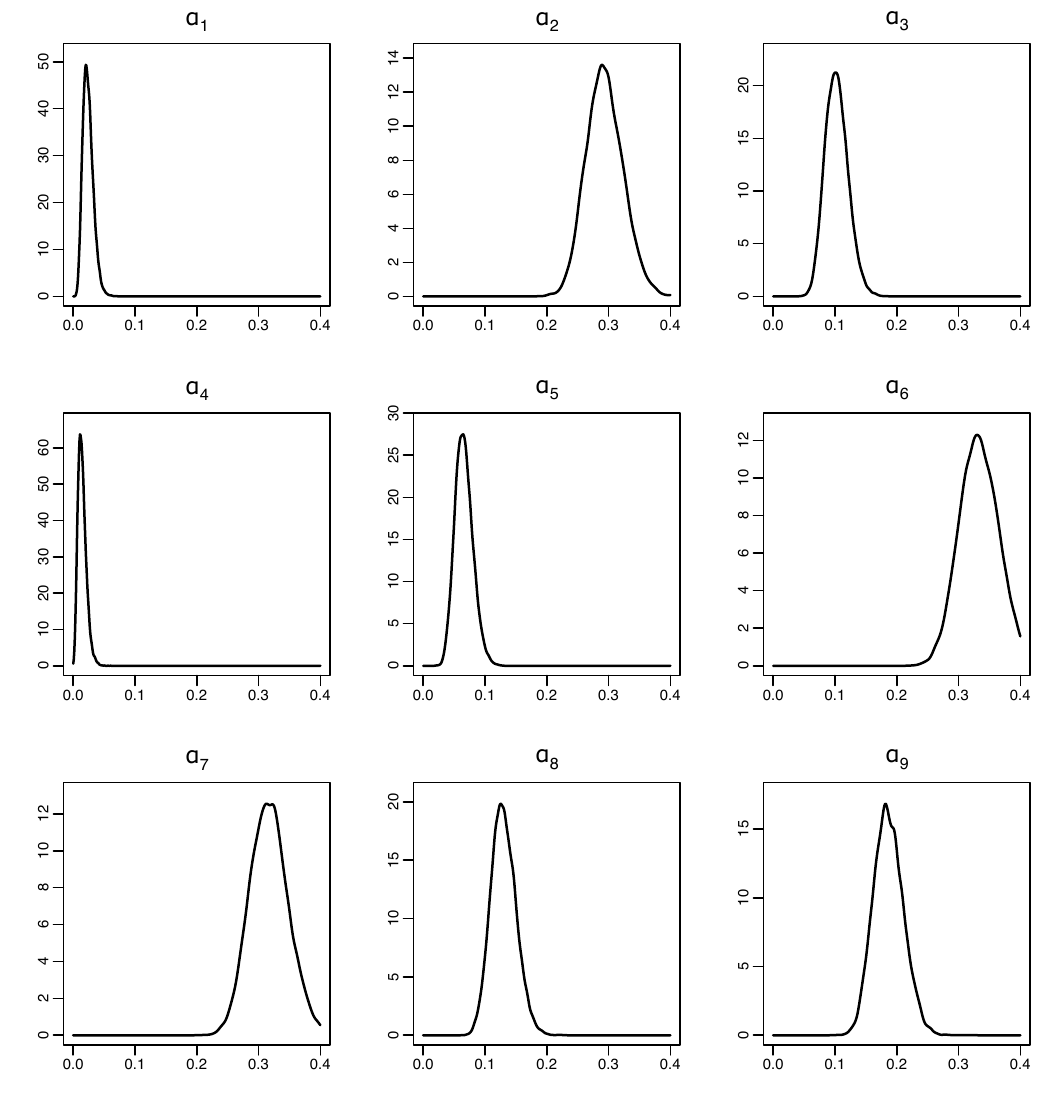}
    \caption{Posterior distributions of the background-rate scale parameters $\alpha_m$ for each participant.} 
    \label{fig_c_scalar}
\end{figure}

Finally, Figure \ref{fig_KLK_comparison} highlights a limitation of the classic Hawkes model. Since it uses only a single influence magnitude for each ordered pair of participants, it cannot distinguish between the effect of immigrant messages, which initiate new conversations, and triggered messages, which are replies within an existing thread. As a result, its posterior distribution often lies between the corresponding Ancestor Hawkes posteriors for $K$ and $L$. The panels in Figure \ref{fig_KLK_comparison} show this clearly: the classic Hawkes estimates tend to average over two qualitatively different influence mechanisms that the Ancestor Hawkes model is able to separate. This provides further evidence that allowing immigrant and triggered messages to have different effects is not merely a modelling convenience, but captures a genuine feature of the group chat dynamics.

\begin{figure}[!htp]
    \centering
    \includegraphics[width=0.7\linewidth]{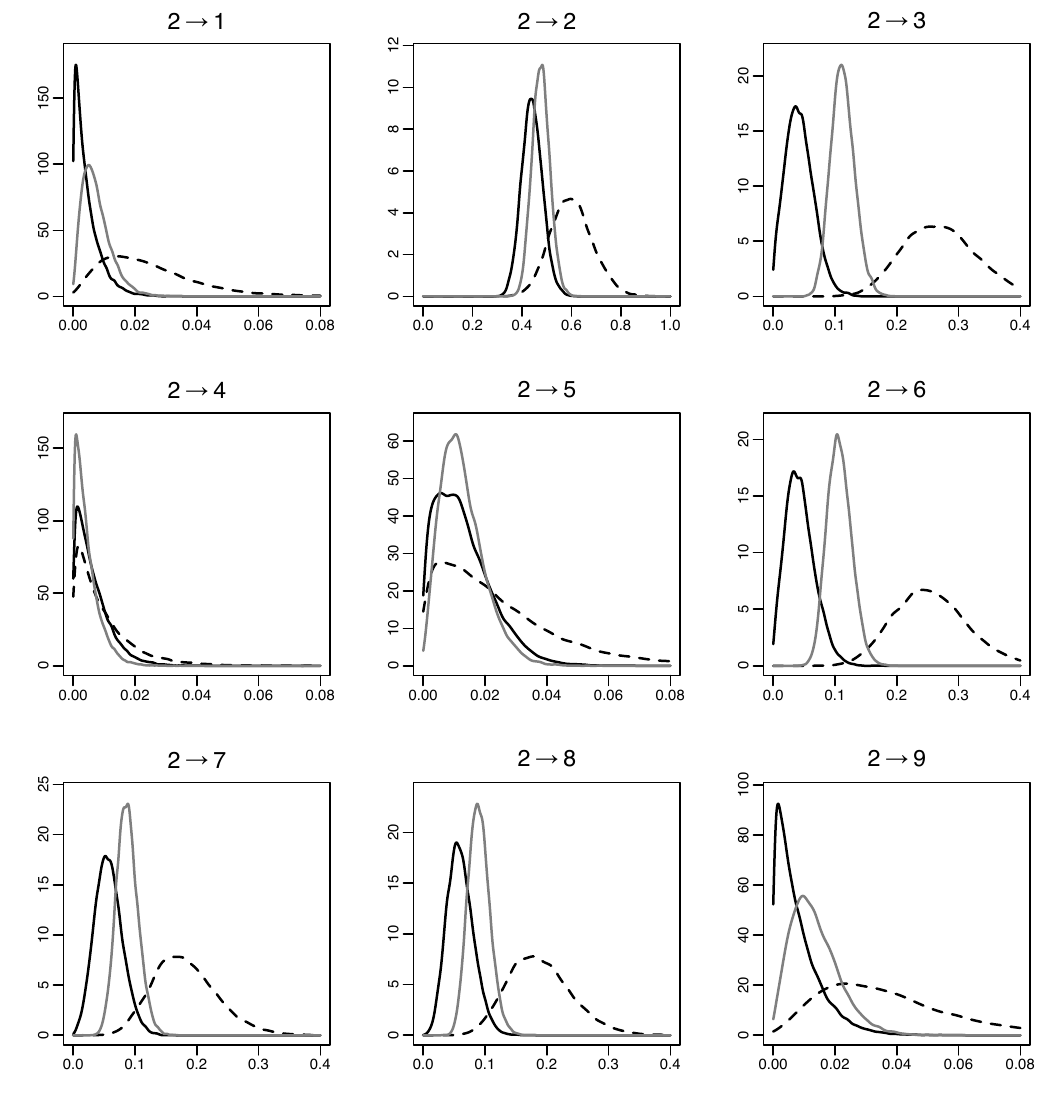}
    \caption{Posterior distribution estimates for the influence magnitude in response to a message  from Participant~$2$. Participants $1$ and $4$ are not shown,  due to their low engagement. The black solid line gives the posterior distribution for $\Kbold$ in the Ancestor Hawkes model and the black dashed line represents posterior estimates for $\Lbold$ in Ancestor Hawkes, the grey line shows the posterior distribution of $\Kbold$ in the classic Hawkes.} 
    \label{fig_KLK_comparison}
\end{figure}

\subsubsection{Summary Statistics} \label{subsec_sumstats_chat}

To test whether the Ancestor Hawkes model is genuinely capturing features of the data that the standard Hawkes model cannot, we analyse the posterior distribution of several summary statistics under both models. This is essentially the standard approach used in posterior predictive checks \cite{gelman1996posterior}. Let $S(\mathbf{Y})$ denote a summary statistic of interest, which computes a scalar quantity from the observed data $\mathbf{Y}$. Let $\boldsymbol{\theta}^{(1)},\ldots,\boldsymbol{\theta}^{(R)}$ denote $R$ samples from the posterior distribution of the parameter vector $\boldsymbol{\theta}$, obtained here via Gibbs sampling. For each $\boldsymbol{\theta}^{(i)}$, we simulate a synthetic data set $\mathbf{Y}^{(i)}$ from the model and compute the corresponding value of the summary statistic, $S(\mathbf{Y}^{(i)})$. The distribution of the $S(\mathbf{Y}^{(i)})$ values is then the posterior predictive distribution of the summary statistic under the model. If the model is a good fit to the data, we would expect the observed value $S(\mathbf{Y})$ to lie close to the centre of this distribution, or at least not too far into its tail. A posterior predictive $p$-value can also be computed from the proportion of the posterior predictive distribution lying above or below $S(\mathbf{Y})$. This closely parallels the frequentist idea of computing a tail probability for a chosen summary statistic under a null model.

In our case, comparing the Ancestor and standard Hawkes models requires choosing summary statistics that capture the clustering structure of the data. This question was considered in \cite{deutsch2021abc}, who study a range of summary statistics in the context of estimating Hawkes processes using Approximate Bayesian Computation. Based on their work, we use the following three statistics:
\begin{enumerate}
\item the mean inter-event time conditional on the inter-event time exceeding the 90th percentile of its distribution,
\item the lag-1 auto-correlation function (ACF) of the inter-event times,
\item the Ripley $K$-statistic for a window of 2 hours, i.e. the average number of events that fall within a 2-hour window after a given event.
\end{enumerate}

Figure~\ref{fig_summary} shows the posterior predictive distributions of these summary statistics under both the Ancestor Hawkes and classic Hawkes models. It can be seen that while the classic Hawkes model does a reasonable job of capturing the auto-correlation in the group chat data, it provides a poorer fit for the other two summary statistics. In particular, it places too much mass on larger values of both the mean upper-tail inter-event time and Ripley’s $K$ statistic. In contrast, the corresponding summary statistics under the Ancestor Hawkes model are much more consistent with the observed data. This suggests that the Ancestor Hawkes model captures aspects of the short-term clustering structure of the group chat that are missed by the classic Hawkes model.

\begin{figure}[!htp]
    \centering
    \includegraphics[width=.85\linewidth]{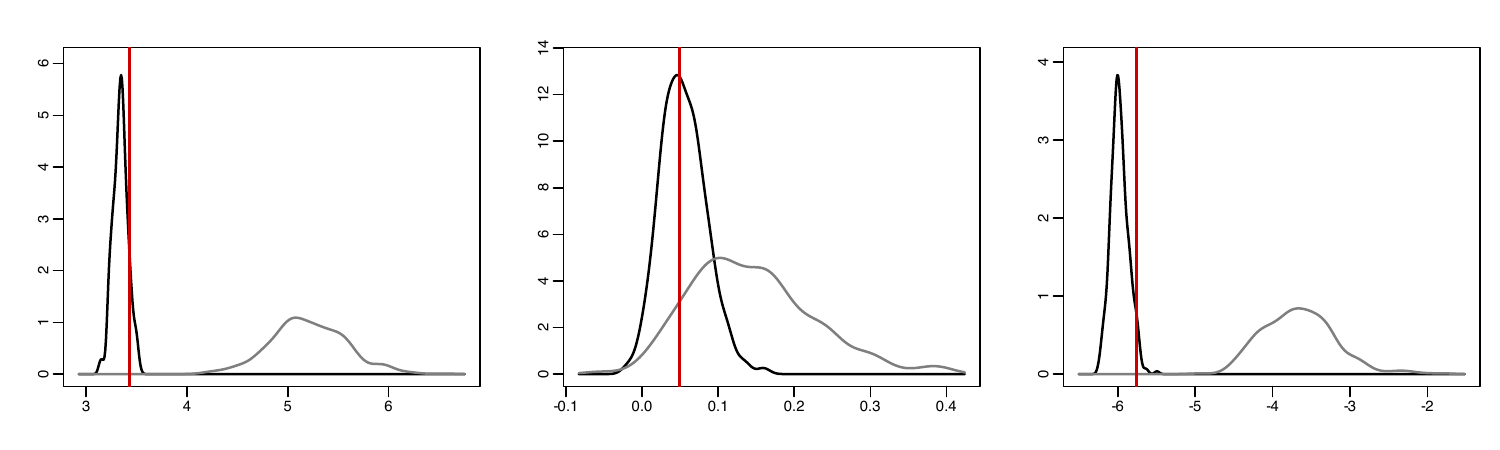}
    \caption{Posterior predictive distributions of three summary statistics for the group chat data. From left to right, the panels show the mean upper-tail inter-event time, the lag-1 autocorrelation of inter-event times, and Ripley's $K$-statistic for a 2-hour window. The black line shows the distribution for data simulated under the Ancestor Hawkes model, while the grey line shows the distribution for data simulated under the classic Hawkes model. The red line shows the observed value of the statistic on the group chat data.}
    \label{fig_summary}
\end{figure}

\begin{figure}[!htp]
    \centering
    \includegraphics[width=.85\linewidth]{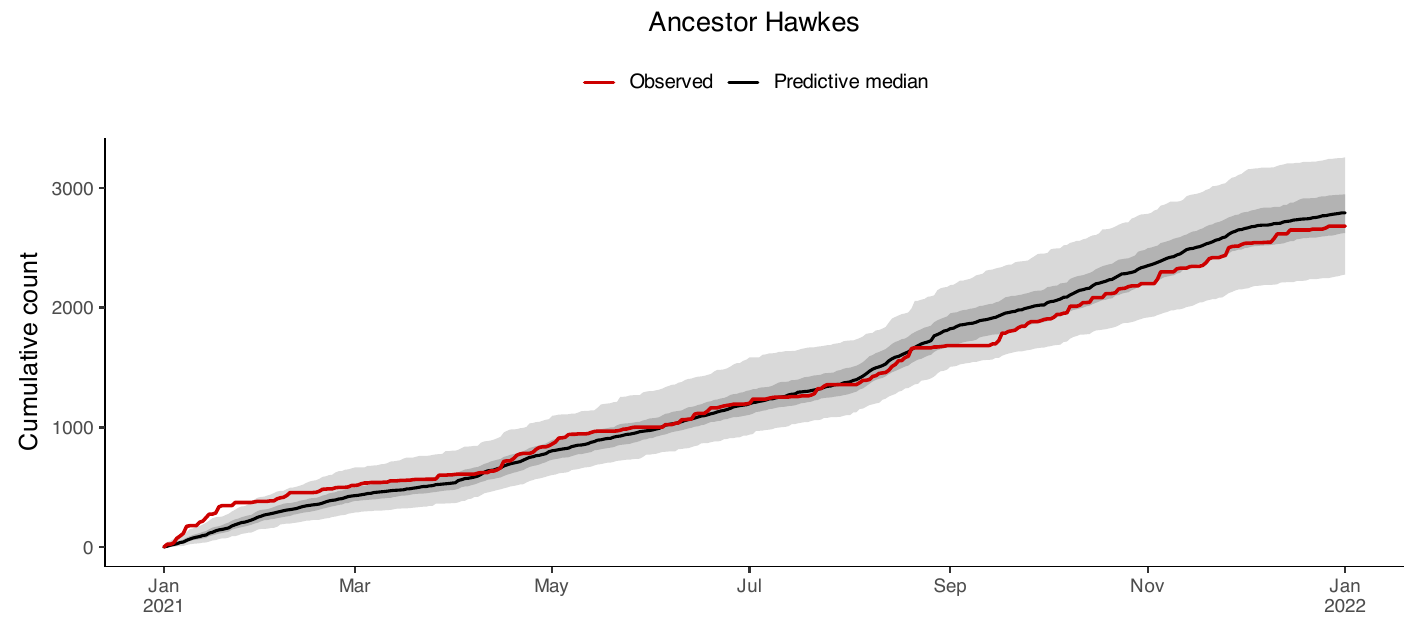}\caption{Posterior predictive cumulative count process for the 2021 group chat data under the Ancestor Hawkes model. The black line shows the posterior predictive median, the grey bands show the 50\% and 95\% posterior predictive intervals, and the red line shows the observed cumulative number of messages over the year.}
    \label{fig_cumulative}
\end{figure}

As one additional diagnostic, Figure~\ref{fig_cumulative} compares the observed cumulative number of messages over 2021 with the posterior predictive cumulative count process under the Ancestor Hawkes model. The observed trajectory lies well within the posterior predictive envelope throughout the year and closely tracks its centre. This suggests that the model is not only capturing local clustering features, as reflected by the summary statistics above, but is also reproducing the broader non-stationary evolution of message activity over calendar time. In particular, the fit appears adequate both in terms of the overall number of events and in how this activity accumulates over the course of the year.

\section{Simulation Study -- Time-varying \texorpdfstring{$\mu_m(t)$}{mu m(t)}}
\label{sec:timevarying}

We conclude with a final simulation study designed to mimic the structure of the real group-chat application. The simulation studies in Section~\ref{sec_simulated_data} used constant background rates and lower-dimensional parameter settings. Here, we instead simulate from a fitted 9-dimensional Ancestor Hawkes model with the same time-varying seasonal background-rate structure used in Section~\ref{sec_groupchat_app}.

Specifically, we take the posterior mean of the parameters estimated from the group-chat data in Section~\ref{sec_groupchat_app}, including the participant-specific background rates, the hour-of-day, day-of-week and month-of-year seasonal effects, the influence matrices \(K\) and \(L\), and the kernel-rate parameters. Treating these posterior means as the generating parameters, we simulate 200 new data sets from the Ancestor Hawkes process over the same observation window. We then refit the Ancestor Hawkes model to each simulated data set using the same priors and MCMC scheme as in the real-data analysis.

This experiment is intended as a recovery check in a setting close to the empirical application. Since the data are generated from the fitted Ancestor Hawkes model, we can compare the estimated influence matrices with the known generating matrices. This directly assesses whether the inference procedure can recover the type of interaction structure interpreted in the group-chat analysis.

Figure~\ref{fig:realdata_sim_recovery} shows the results. The left and middle columns compare the generating influence matrices with the average posterior mean estimates obtained after refitting the model to the 200 simulated data sets. The top row corresponds to the immigrant-event influence matrix \(K\), while the bottom row corresponds to the triggered-event influence matrix \(L\). The recovered matrices closely match the generating matrices, including the main high-influence entries and the broader pattern of heterogeneity across participant pairs.

\begin{figure}
    \centering
    \includegraphics[width=\linewidth]{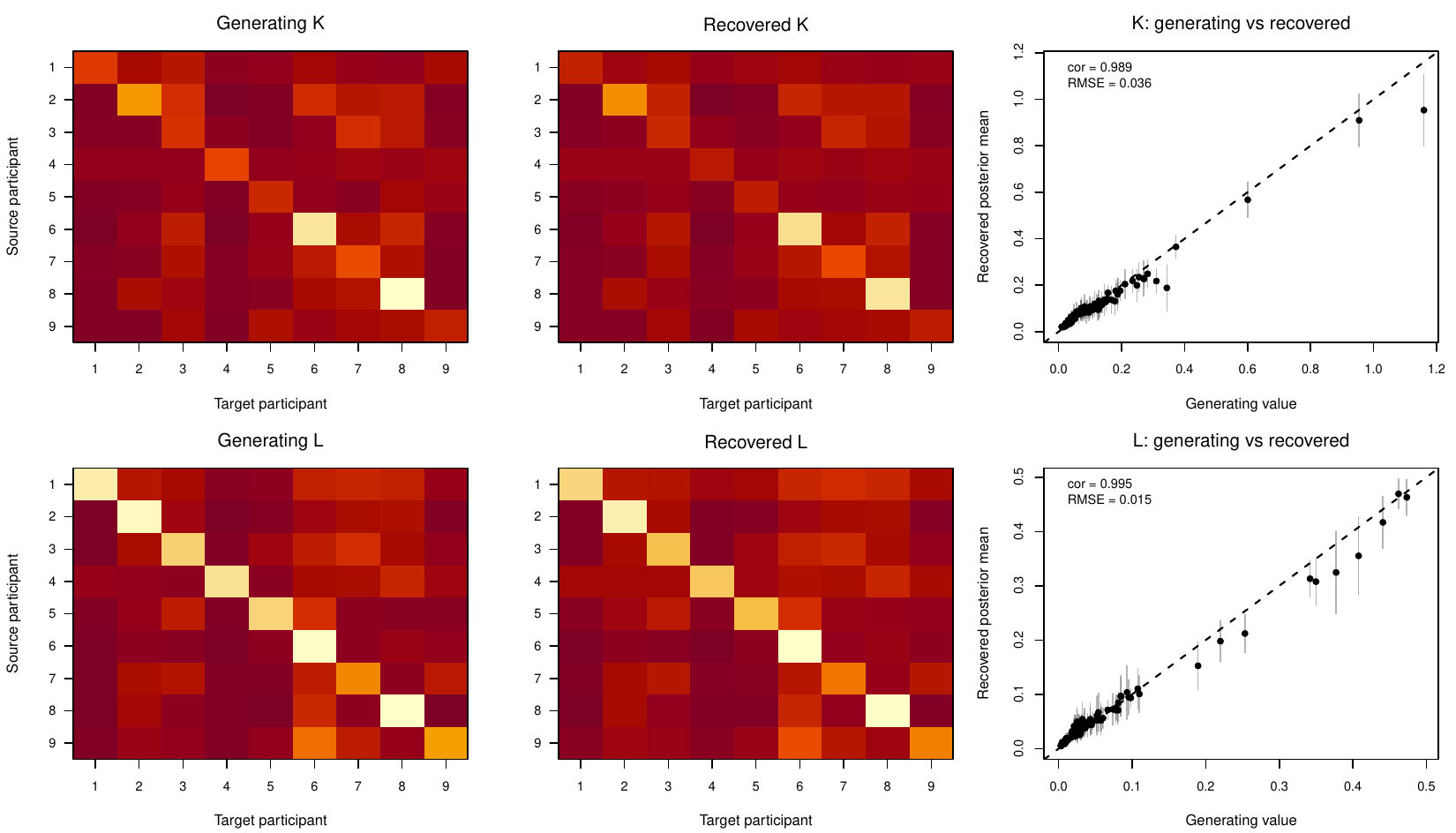}
    \caption{Real-data-like simulation study. Data were simulated from the Ancestor Hawkes model using posterior mean parameters from the group-chat application, including the fitted seasonal background rate. The left and middle columns compare the generating influence matrices with the average posterior mean estimates obtained after refitting the Ancestor Hawkes model to 200 simulated data sets. The right column plots generating values against recovered posterior means for all ordered participant pairs; vertical bars show one standard deviation of the posterior mean estimates across the 200 simulated data sets. The top row corresponds to immigrant-event influence \(K\), and the bottom row to triggered-event influence \(L\).}
    \label{fig:realdata_sim_recovery}
\end{figure}

The right column of Figure~\ref{fig:realdata_sim_recovery} gives a more direct entrywise comparison. Each point corresponds to one ordered participant pair, with the horizontal coordinate giving the generating value and the vertical coordinate giving the average recovered posterior mean. Vertical bars show one standard deviation of the recovered posterior means across the 200 simulated data sets. The entrywise correlation between the generating and recovered influence matrices is \(0.989\) for \(K\) and \(0.995\) for \(L\), with RMSEs \(0.036\) and \(0.015\), respectively. The largest entries of \(K\) show some mild shrinkage, consistent with the shrinkage prior also observed in the constant-background simulations, but the overall recovery is strong.

Overall, this simulation supports the interpretation of the fitted group-chat model in Section~\ref{sec_groupchat_app}. In a real-data-like setting with 9 dimensions, a time-varying seasonal background rate, and influence parameters taken from the empirical analysis, the corrected Ancestor Hawkes sampler is able to recover the main immigrant- and triggered-event influence structures.

\section{Discussion} \label{sec_discussion_anc}

In this paper, we introduced the Ancestor Hawkes process, a multivariate Hawkes model in which the influence of an event depends not only on the dimension in which it occurred, but also on whether it was an immigrant or a triggered event. This provides a natural extension of the classic Hawkes process in settings where conversation-initiating events and follow-up events are expected to play qualitatively different roles. The model retains a clear branching-process interpretation while allowing a richer description of local interaction dynamics.

Our simulation study showed that when immigrant and triggered events have genuinely different effects, the classic Hawkes model tends to blur these effects together, whereas the Ancestor Hawkes model is able to recover them. In the group chat application, this distinction revealed meaningful differences in reply behavior that cannot be expressed as cleanly under the classic Hawkes process. Importantly, this was achieved using only the times at which messages were sent and the identity of the sender, without requiring access to message content. In this sense, the proposed model offers a method for studying response preferences and interaction structure in conversational data.

There are several natural directions for future work. Additional covariates could be incorporated into the model, particularly in messaging applications where explicit reply information is available and could help inform the latent branching structure. It would also be natural to allow the influence of an event to depend on message type, for example whether it is text, an image, a GIF, or an emoji-only message. More broadly, it would be interesting to develop faster estimation procedures and to study these models in larger multivariate settings where the distinction between exogenous starts and reply-driven propagation may be especially important.

\section*{Data availability}

The anonymised group-chat data used in this paper, consisting only of message timestamps and sender identifiers, are available at \url{https://github.com/isabella-de/AncestorHawkes}. No message content is included.

\section*{Statements and Declarations}

\textbf{Funding.} The authors received no specific funding for this work.

\textbf{Competing interests.} The authors have no competing interests to declare that are relevant to the content of this article.

\textbf{Ethics approval.} Ethical approval for the analysis of the group-chat data was obtained through the authors' host institution.

\textbf{Consent to participate.} All participants in the group chat gave consent for their anonymised timestamp/sender data to be analysed.

\clearpage
\appendix

\section{MCMC Diagnostics}
\label{app:mcmc-diagnostics}

\begin{figure}[!htp]
    \centering
    \includegraphics[width=.85\linewidth]{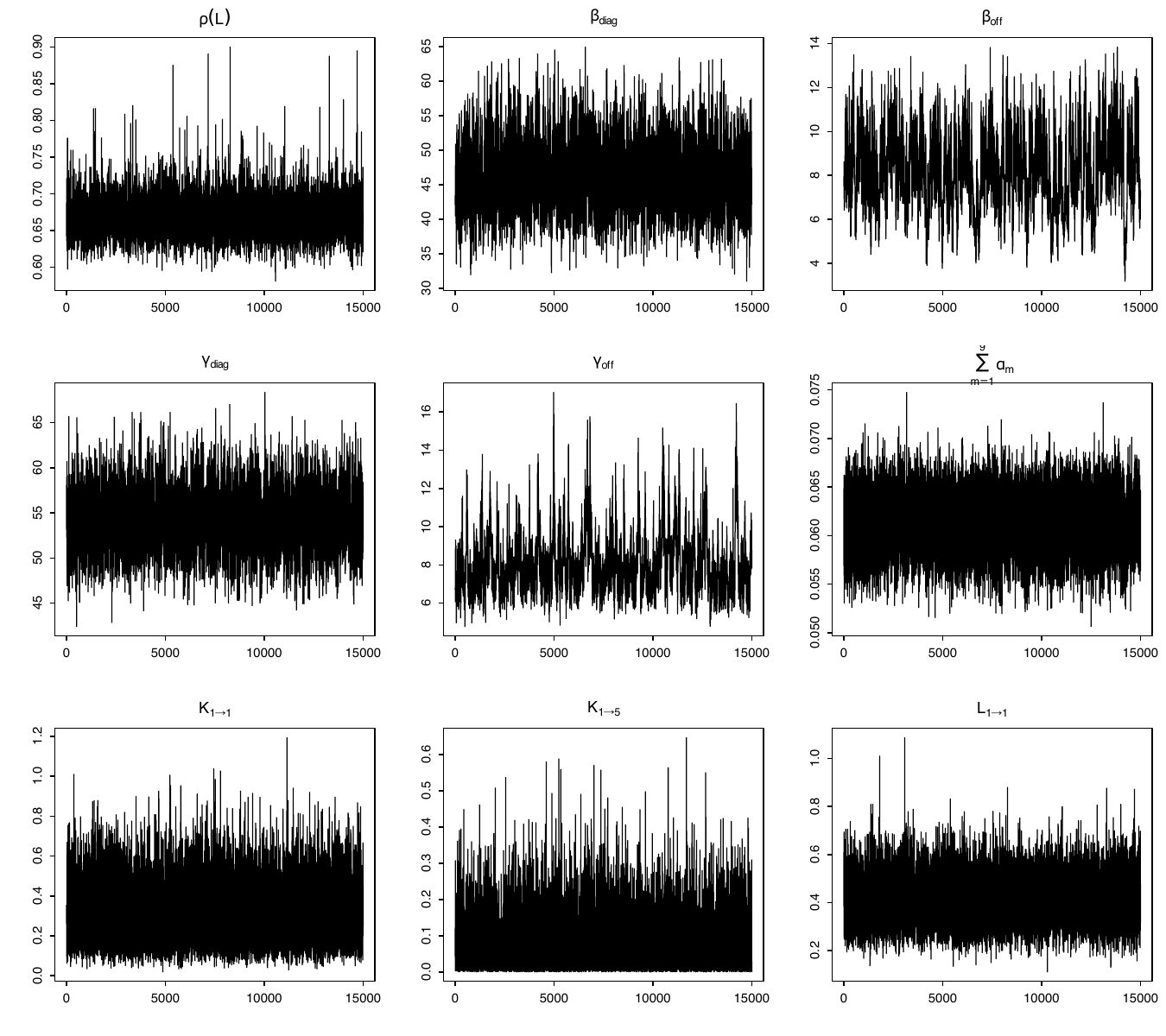}
    \caption{MCMC trace plots for selected posterior quantities in the group-chat data. 
    Here, $\rho(L)$ denotes the spectral radius of $L$, which must be less than 1 
    for stability.}
    \label{fig:mcmc-diagnostics}
\end{figure}

\clearpage
\bibliography{references}

@article{schatz2022arma,
  author  = {Schatz, Michael and Wheatley, Spencer and Sornette, Didier},
  title   = {The {ARMA} Point Process and its Estimation},
  journal = {Econometrics and Statistics},
  volume  = {24},
  pages   = {164--182},
  year    = {2022},
  doi     = {10.1016/j.ecosta.2021.11.002}
}

@InProceedings{achab2017,
  author    = {Achab, Massil and Bacry, Emmanuel and Ga{\"i}ffas, St{\'e}phane
               and Mastromatteo, Iacopo and Muzy, Jean-Fran{\c{c}}ois},
  title     = {Uncovering Causality from Multivariate Hawkes Integrated Cumulants},
  booktitle = {Proceedings of the 34th International Conference on Machine Learning},
  editor    = {Precup, Doina and Teh, Yee Whye},
  series    = {Proceedings of Machine Learning Research},
  volume    = {70},
  pages     = {1--10},
  publisher = {PMLR},
  address   = {Sydney, Australia},
  year      = {2017},
}

@article{deutsch2025cannibalisation,
  author  = {Deutsch, Isabella and Ross, Gordon J.},
  title   = {Estimating Product Cannibalisation in Wholesale using Multivariate Hawkes Processes with Inhibition},
  journal = {The Annals of Applied Statistics},
  volume  = {19},
  number  = {1},
  pages   = {235--260},
  year    = {2025},
}

@article{gelman1996posterior,
  author  = {Gelman, Andrew and Meng, Xiao-Li and Stern, Hal},
  title   = {Posterior Predictive Assessment of Model Fitness via Realized Discrepancies},
  journal = {Statistica Sinica},
  volume  = {6},
  number  = {4},
  pages   = {733--760},
  year    = {1996}
}

@misc{deutsch2021abc,
  author       = {Deutsch, Isabella and Ross, Gordon J.},
  title        = {ABC Learning of Hawkes Processes with Missing or Noisy Event Times},
  year         = {2021},
  eprint       = {2006.09015},
  archivePrefix= {arXiv},
  primaryClass = {stat.AP},
  note         = {arXiv:2006.09015}
}

@article{Pitkin2024,
	title = {Bayesian hierarchical modelling of sparse count processes in retail analytics},
	journal = {Annals of Applied Statistics},
	author = {Pitkin, James and Manoloupoulou, Ioanna and Ross, Gordon},
	year = {2024},
}

@Article{Kolev2019,
author="Kolev, Aleksandar A.
and Ross, Gordon J.",
title="Inference for {ETAS} models with non-{P}oissonian mainshock arrival times",
journal="Statistics and Computing",
year="2019",
volume="29",
number="5",
pages="915--931"
}

@article{Veen_2008_estimation,
 author = {Alejandro Veen and Frederic P. Schoenberg},
 journal = {Journal of the American Statistical Association},
 number = {482},
 pages = {614--624},
 publisher = {[American Statistical Association, Taylor & Francis, Ltd.]},
 title = {Estimation of Space-Time Branching Process Models in Seismology Using an {EM}-Type Algorithm},
 volume = {103},
 year = {2008}
}

@article{rasmussen_bayesian_2013,
	title = {Bayesian Inference for {Hawkes} Processes},
	volume = {15},
	number = {3},
	journal = {Methodology and Computing in Applied Probability},
	author = {Rasmussen, Jakob Gulddahl},
	year = {2013},
	pages = {623--642}
}

@article{hawkes_spectra_1971,
	title = {Spectra of some self-exciting and mutually exciting point processes},
	volume = {58},
	number = {1},
	journal = {Biometrika},
	author = {Hawkes, Alan G.},
	year = {1971},
	pages = {83--90}
}

@article{rizoiu_tutorial_2017,
	title = {A Tutorial on {Hawkes} Processes for Events in Social Media},
	journal = {arXiv},
	author = {Rizoiu, Marian-Andrei and Lee, Young and Mishra, Swapnil and Xie, Lexing},
	year = {2017}
}

@article{ogata_statistical_1988,
	title = {Statistical Models for Earthquake Occurrences and Residual Analysis for Point Processes},
	volume = {83},
	number = {401},
	journal = {Journal of the American Statistical Association},
	author = {Ogata, Yosihiko},
	year = {1988},
	pages = {9--27}
}

@article{hawkes_cluster_1974,
 author = {Alan G. Hawkes and David Oakes},
 journal = {Journal of Applied Probability},
 number = {3},
 pages = {493--503},
 title = {A Cluster Process Representation of a Self-Exciting Process},
 volume = {11},
 year = {1974}
}

@inproceedings{miscouridou_modelling_2018,
  author    = {Miscouridou, Xenia and Caron, Fran{\c{c}}ois and Teh, Yee Whye},
  title     = {Modelling sparsity, heterogeneity, reciprocity and community structure
               in temporal interaction data},
  booktitle = {Advances in Neural Information Processing Systems 31},
  editor    = {Bengio, S. and Wallach, H. and Larochelle, H. and Grauman, K.
               and Cesa-Bianchi, N. and Garnett, R.},
  pages     = {2349--2358},
  publisher = {Curran Associates, Inc.},
  address   = {Red Hook, NY, USA},
  year      = {2018},
}

@Preamble{ " \newcommand{\noop}[1]{} " }

@article{molkenthin_gp-etas_2022,
	title = {{GP}-{ETAS}: semiparametric {Bayesian} inference for the spatio-temporal epidemic type aftershock sequence model},
	volume = {32},
	number = {2},
	urldate = {2022-08-03},
	journal = {Statistics and Computing},
	author = {Molkenthin, Christian and Donner, Christian and Reich, Sebastian and Zöller, Gert and Hainzl, Sebastian and Holschneider, Matthias and Opper, Manfred},
	year = {2022},
	pages = {29},
}

@article{bacry_hawkes_2015,
	title = {Hawkes {Processes} in {Finance}},
	volume = {01},
	number = {01},
	journal = {Market Microstructure and Liquidity},
	author = {Bacry, Emmanuel and Mastromatteo, Iacopo and Muzy, Jean-François},
	year = {2015},
	pages = {1550005}
}

@article{mohler2013modeling,
  title={Modeling and estimation of multi-source clustering in crime and security data},
  author={Mohler, George},
  volume = {7},
  number = {3},
  journal={The Annals of Applied Statistics},
  pages={1525--1539},
  year={2013},
}

@article{kolev_semiparametric_2020,
    author = {Gordon J. Ross and Aleksandar A. Kolev},
    title = {Semiparametric {Bayesian} forecasting of SpatioTemporal earthquake occurrences},
    volume = {16},
    journal = {The Annals of Applied Statistics},
    number = {4},
    pages = {2083 -- 2100},
    year = {2022}
}

@article{ross_bayesian_2021,
	title = {Bayesian Estimation of the {ETAS} Model for Earthquake Occurrences},
	volume = {111},
	number = {3},
	journal = {Bulletin of the Seismological Society of America},
	author = {Ross, Gordon J.},
	year = {2021},
	pages = {1473--1480}
}

@article{guo_who_2019,
	title = {Who is answering whom? {Finding} “Reply-To” relations in group chats with deep bidirectional {LSTM} networks},
	volume = {22},
	journal = {Cluster Computing},
	author = {Guo, Gaoyang and Wang, Chaokun and Chen, Jun and Ge, Pengcheng and Chen, Weijun},
	year = {2019},
	pages = {2089--2100},
}

@inproceedings{li_brunc_2020,
  author    = {Hui Li and Hui Li and Sourav S. Bhowmick},
  title     = {{BRUNCH}: Branching Structure Inference of Hybrid Multivariate Hawkes
               Processes with Application to Social Media},
  booktitle = {Advances in Knowledge Discovery and Data Mining -- 24th Pacific-Asia
               Conference, {PAKDD} 2020, Singapore, May 11--14, 2020,
               Proceedings, Part {I}},
  series    = {Lecture Notes in Computer Science},
  volume    = {12084},
  pages     = {553--566},
  publisher = {Springer},
  address   = {Cham},
  year      = {2020},
}

@article{mannell_plural_2020,
  title={Plural and porous: reconceptualizing the boundaries of mobile messaging group chats},
  author={Mannell, Kate},
  journal={Journal of Computer-Mediated Communication},
  volume={25},
  number={4},
  pages={274--290},
  year={2020}
}

@phdthesis{markwick_bayesian_2020,
	title = {Bayesian {Nonparametric} {Hawkes} {Processes} with {Applications}},
	school = {University College London},
	author = {Markwick, Dean},
	year = {2020}
}

\end{document}